\newif\ifarx \arxtrue
\ifarx
\documentclass[a4paper,DIV=15,10pt]{scrartcl}
\else
\documentclass[english]{jssst_ppl}
\fi
\usepackage{mysettings}
\usepackage{mysymbols}
\usepackage{docmute}

\begin{document}
%#!latexmk -c -gg -lualatex main.tex
\title{A study for recovering the cut-elimination property in cyclic proof systems by restricting the arity of inductive predicates}
\ifarx
\author{
  Yukihiro Oda\thanks{Tohoku University, \texttt{yukihiro3socrates6hilbert [at] gmail.com}}
  \and
  Daisuke Kimura\thanks{Toho University, \texttt{kmr [at] is.sci.toho-u.ac.jp}}
}
\else
\author{
  Yukihiro Masuoka$^1$
  and
  Daisuke Kimura$^2$
}
%\date{\today}

\inst{%
  \begin{tabular}{ll}
    %  $^1$ Department of Informatics, The Graduate University for Advanced Studies
    $^1$ The Graduate University for Advanced Studies
    &
    %    \texttt{yukihiro\_m[at]nii.ac.jp}
    \texttt{yukihiro\_m@nii.ac.jp}
    \\
    $^2$ Toho University
    &
    \texttt{kmr@is.sci.toho-u.ac.jp}
  \end{tabular}
}
\fi
\maketitle

\begin{abstract}
 The framework of cyclic proof systems provides a reasonable proof system for logics with inductive definitions.
It also offers an effective automated proof search procedure for such logics without finding induction hypotheses.
Recent researches have shown that the cut-elimination property, one of the most fundamental properties in proof theory, of cyclic proof systems for several logics does not hold.
These results suggest that a naive proof search, which avoids the Cut rule, is not enough.

This paper shows that the cut-elimination property still fails in a simple cyclic proof system even if we restrict languages to unary inductive predicates and unary functions, aiming to clarify why the cut-elimination property fails in the cyclic proof systems. The result in this paper is a sharper one than that of the first authors' previous result, which gave a counterexample using two ternary inductive predicates and a unary function symbol to show the failure of the cut-elimination property in the cyclic proof system of the first-order logic.

%A cyclic proof system is a proof system
%whose proof figure is a tree with cycles.
%The cut-elimination in a proof system is fundamental.

\end{abstract}

%#!latexmk -c -gg -lualatex main.tex
\section{Introduction}

Inductive definition is a way to define mathematical objects based on the induction principle.
Several notions, which are essential in both mathematics and computer science,
such as natural numbers, lists, and binary trees, are inductively defined. 
Inductively defined predicates are called \textit{inductive predicates}.
A typical example of inductive predicates is $\N{x}$ that means ``$x$ is a natural number''. 
It is given by the following Martin L\"{o}f style schemata (called \textit{productions})~\cite{MarinLof1971}:
\begin{center}
  $\begin{inlineprooftree}
    \AxiomC{}
    \UnaryInfC{$\N{\zero}$}
  \end{inlineprooftree}$
  \quad
  and
  \quad
  $\begin{inlineprooftree}
    \AxiomC{$\N{x}$}
    \UnaryInfC{$\N{\suc{x}}$}
  \end{inlineprooftree}$, 
\end{center}
where $\zero$ is a constant symbol, and $\sucsy$ is a unary function symbol.
The first production means that $\N{\zero}$ holds without any assumptions, namely,
it says that ``$\zero$ is a natural number''.
The second one means that $\N{x}$ implies $\N{\suc{x}}$, namely,
it says that ``if $x$ is a natural number, then $\suc{x}$ is also a natural number''.
It is also assumed that no other rule can be applied to obtain $\N{t}$ for any term $t$.
Hence, these productions say that the predicate $\Nsy$ is defined as the least one
that satisfies the following equivalence: 
\[
\N{x} \Leftrightarrow (x = \zero \lor \exists y.(x = \suc{y} \land \N{y})). 
\]
This equivalence gives an alternative definition of the natural number predicate $\Nsy$
instead of giving its productions. 

Proof systems for logics with inductive predicates have been studied in the literature~\cite{MarinLof1971,Tiu2012}.
It is known that Gentzen's sequent calculus LK can be extended with inductive predicates in a uniform way. 
For example, the following inference rules for the natural number predicate $\Nsy$ are generated
from the productions of $\Nsy$:
\begin{center}
  \begin{inlineprooftree}\small
    \AxiomC{}
    \UnaryInfC{$\Gamma \vdash \Delta, \N{\zero}$}
  \end{inlineprooftree}, 
  \begin{inlineprooftree}\small
    \AxiomC{$\Gamma \vdash \Delta, \N{t}$}
    \UnaryInfC{$\Gamma \vdash \Delta, \N{\suc{t}}$}    
  \end{inlineprooftree}, 
  \begin{inlineprooftree}\small
    \AxiomC{$\Gamma \vdash \Delta, F[\zero]$}
    \AxiomC{$\Gamma, F[x] \vdash \Delta, F[\suc{x}]$}
    \AxiomC{$\Gamma, F[t] \vdash \Delta$}
    \RightLabel{(IND)}
    \TrinaryInfC{$\Gamma,\N{t} \vdash \Delta$}
  \end{inlineprooftree}, 
\end{center}
where $\Gamma$ and $\Delta$ are multisets of formulas, $x$ is a fresh variable, $t$ is a term,
$F$ is a formula with a fixed variable $z$, and $F[t]$ is the result of substituting $t$ for $z$ in $F$. 
The last inference rule (IND) corresponds to the induction principle on the natural numbers,
and the formula $F$ is its induction hypothesis. 
Although the last rule is a reasonable formalization of the induction principle, 
it causes a difficult problem when we apply a naive proof search algorithm to this proof system
because we need to find (or guess) an appropriate induction hypothesis $F$ in the upper sequents
from the lower sequent. 

An alternative choice to formulate inductive predicates in sequent calculi is to adopt rules,
instead of the (IND) rules,  that unfold inductive predicates on the left-hand side of a sequent,
according to the equivalence that defines the predicate. 
Brotherston and Simpson~\cite{Brotherston2011} proposed the proof systems (called $\LKIDOm$ and $\CLKID$)
for classical first-order logic, based on this idea. 
In their systems, the unfolding rules (called the \textit{casesplit} rule)
for the natural number predicate $\Nsy$ is given as follows:
\begin{center}
  \begin{inlineprooftree}
    \AxiomC{$\Gamma, t=\zero \vdash \Delta$}
    \AxiomC{$\Gamma, t=\suc{y}, \N{y} \vdash \Delta$}
    \RightLabel{$\ruleCase{\Nsy}$}
    \BinaryInfC{$\Gamma,\N{t} \vdash \Delta$}
  \end{inlineprooftree},
  where $y$ is a fresh variable. 
\end{center}
Although this idea gives a solution to avoid the problem of finding induction hypotheses in proof search,
it requires considering infinite proofs, which might have infinite paths. 
The system $\LKIDOm$ is a proof system that admits such infinite proofs which satisfy
a condition (called the \textit{global trace condition}) that ensures the soundness of the system.
%This condition requires each infinite path to contain a sequence, called a \textit{trace},
%of atomic formulas with inductive predicates in which the inductive predicates are unfolded infinitary many times.
The cyclic proof system $\CLKID$ is a reasonable restriction of $\LKIDOm$
that admits only proofs which are regular trees.
It is formulated as a finite derivation tree with open assumptions (called \textit{buds}) and additional edges
that connects from each bud to an internal node (called \textit{companion} of the bud). 

The framework of cyclic proof systems gives a general way to formalize logics with inductive definitions
(or the least/greatest fixed point operators). 
It has been proposed several cyclic proof systems other than classical first-order logic, 
such as logic of bunched implications~\cite{Brotherston2007}, separation logic~\cite{Brotherston11b}, and
linear logic with the fixed point operators~\cite{Baelde2016,DoumanePhD}. 

The cut-elimination property of a proof system states that any provable sequent in the system is
also provable without the rule $\ruleCut$, which is given below:
\begin{center}
  \begin{inlineprooftree}
    \AxiomC{$\Gamma \vdash \Delta, F$}
    \AxiomC{$F, \Gamma \vdash \Delta$}
    \RightLabel{$\ruleCut$}
    \BinaryInfC{$\Gamma \vdash \Delta$}
  \end{inlineprooftree}. 
\end{center}
This property is one of the most fundamental properties of proof systems
because it helps us to investigate a given proof system. 
For example, some important properties, such as the subformula property and consistency of the proof system,
are often obtained using the cut-elimination property. 
It is also important for the proof search method
since it ensures that it is enough to search except for $\ruleCut$,
which requires to find an appropriate cut formula from possibly infinitary many candidates. 

Some infinite proof systems are known to enjoy the cut-elimination property:
Brotherston and Simpson proved that the cut-elimination theorem of $\LKIDOm$ by showing $\LKIDOm$ is sound
and cut-free complete to a standard model for inductive predicates~\cite{Brotherston2011}.
Fortier and Santocanale~\cite{Fortier2013} introduced a cyclic proof system for additive linear logic
with the least and greatest fixed point operators, and showed that the rule $\ruleCut$ can be eliminated
if we admit to lose the regularity of proof-trees. Doumane~\cite{DoumanePhD} investigated
an infinite proof system $\mu\mathtt{MALL}^\infty$ for the multiplicative and additive linear logic
with the least and greatest fixed point operators,
and showed its cut-elimination theorem. 

In contrast, the situation for cyclic proof systems is totally different. 
The open problem about the cut-elimination property of $\CLKID$ by Brotherston
was negatively solved in the first authors' recent work~\cite{Masuoka2021}.
The second author showed that the cut-elimination property does not hold in a cyclic proof system
of separation logic~\cite{Kimura2020}. 
Saotome showed the failure of the cut-elimination property for the cyclic proof system of the logic of bunched implications
even if we restrict inductive predicates to nullary predicates~\cite{Saotome2021}. 

There are automated theorem provers based on proof search algorithms
of cyclic proof systems~\cite{Brotherston11b,Brotherston12,Chu15,Songbird1,Songbird2,Tatsuta18}. 
Some of them adopt an additional mechanism that guesses cut formulas.
They record sequents which are found during an execution of proof search procedure, and then try to generate possible candidates of
cut formulas from the recorded sequents~\cite{Chu15,Cyclist,Songbird1,Songbird2}. 
This technique extends the ability of the provers to find cyclic proofs which may contain $\ruleCut$, 
and also gives an efficient proof search procedure.
However, Saotome~\cite{Saotome2020} suggested that there still exist sequents in 
the symbolic heap separation logic that cannot be found by a normal proof search procedure
admitting $\ruleCut$ whose cut formulas are presumable from the goal sequent. 
Recently the framework of cyclic proof-search has been studied from the viewpoint of software verification~\cite{Tellez2020,Tsukada2022}. 

In our recent research, we have fixed our attention to the following two natural questions about the cut-elimination property in cyclic proofs. 
\begin{itemize}
\item
  Can we recover the cut-elimination property of cyclic proof systems (in particular $\CLKID$)
  by restricting the definitions of inductive predicates?
\item
  Is there a reasonable restriction of $\ruleCut$ whose cut formulas can be found in small search space, 
  and that does not lose provability of the original cyclic proof system. 
\end{itemize}

This paper focuses on the first question.
We show that the cut-elimination property of $\CLKID$ still fails
even if we restrict inductive predicates to two unary predicates ($\FsTsy$ and $\TeFsy$) with two constant symbols ($\s$ and $\e$) and 
one unary function symbol ($\mathnext$). 
The proof technique of this paper is a modified and simplified one of the first author's previous work~\cite{Masuoka2021}.
Our discussion starts from introducing a simple subsystem (called the base system $\Base$) that only contains equality and inductive predicates
because other logical connectives and quantifiers are not necessary. 
Then we define a cyclic proof system (called $\CBaseOm$) of $\Base$, and apply the proof technique to
our counterexample $\TeF{\s} \vdash \FsT{\e}$. 
This counterexample also works for $\CLKID$ since we can easily check that
if $\TeF{\s} \vdash \FsT{\e}$ is cut-free provable in $\CLKID$, then its proof is also a cut-free proof in $\CBaseOm$.

The reminder of this paper is structured as follows. 
\cref{sec:syntax} introduces the base system $\Base$.
\cref{sec:proof-system} defines the cyclic proof system $\CBaseOm$ of $\Base$.
In \cref{sec:main} we show the main theorem by giving our counterexample for the cut-elimination property of $\CBaseOm$.
\cref{sec:conclusion} concludes.

%#!latexmk -c -gg -lualatex main.tex
\section{The base system}\label{sec:syntax}

In this section, we present the base system $\Base$, which only contains equations and inductive definitions.
It is a subsystem of the first-order logic with inductive definitions $\mathrm{FOL}_\mathrm{ID}$~\cite{Brotherston2011}. 
The system $\Base$ is developed from a language that consists of
countable number of variable symbols (denoted by $x,y,z,\ldots$)
and arbitrary number of function symbols (denoted by $f$),
finite number of \emph{ordinary predicate symbols} (denoted by $Q_1,\ldots,Q_m$),
and finite number of \emph{inductive predicate symbols} (denoted by $P_1,\ldots,P_n$),
where $f$, $Q_i$, and $P_j$ have their own arities $\arity{f}$, $\arity{Q_i}$, and $\arity{P_j}$, respectively.
A function symbol with arity zero is called a constant symbol. 
We use $R$ for a meta variable that ranges over both ordinary predicate symbols and inductive predicate symbols. 

\emph{Terms} (denoted by $t$ and $u$) for $\Base$ are defined by 
\[
t \BNFeq x \mathrelbar f\overbrace{t\cdots t}^{\text{$n$}},
\]
where $n = \arity{f}$.

For a unary function symbol $f$, we use an abbreviation $f^{n}t$ for $f\cdots ft$ ($n$ times of $f$). 

We write $\mathvect{x}$ for a sequence of variables and
$\mathvect{t}$ for a sequence of terms. 
We also write $\mathvect{t}\mleft(\mathvect{x}\mright)$ for $\mathvect{t}$
in which the variables $\mathvect{x}$ occur.
The length of a sequence $\mathvect{t}$ is written $|\mathvect{t}|$. 

A \emph{formula} (denoted by $\varphi$) of $\Base$ is defined as follows:
\[
\varphi \BNFeq t = t \mathrelbar R(\mathvect{t}),
\]
where $|\mathvect{t}| = \arity{R}$. 
We define \emph{free variables} as usual, and
$\mathFV{\varphi}$ is defined as the set of free variables in $\varphi$.

%% \emph{Formulas} are defined by
%% \[
%% \varphi \BNFeq A \mathrelbar \lnot\varphi \mathrelbar  \varphi \mathrel{\land} \varphi \mathrelbar  \varphi \mathrel{\lor} \varphi \mathrelbar \mathop{\exists x} \varphi \mathrelbar \mathop{\forall x} \varphi,
%% \]
%% where $A$ is an atomic formula and $x$ is a variable.

We write
$\varphi\mleft[x_{0}:=t_{0}, \dots, x_{r}:=t_{r}\mright]$ for a formula obtained from a formula $\varphi$
by simultaneously substituting terms $t_{0}$, $\ldots$,  $t_{r}$ for variables $x_{0}$, $\ldots$, $x_{r}$,
respectively.
We sometimes write $\theta$ for $x_{0}:=t_{0}, \dots, x_{r}:=t_{r}$.

Inductive predicate symbols are given with an \emph{inductive definition set}, which is defined as follows. 
\begin{definition}[Inductive definition set]
  A \emph{production} for $P_j$ is defined as
  \begin{center}
    \begin{inlineprooftree}
      \AxiomC{ $Q_{1}(\mathvect{u}_{1}) \quad \cdots \quad Q_{h}(\mathvect{u}_{h}) \quad P_{j_{1}}(\mathvect{t}_{1}) \quad \cdots \quad P_{j_{m}}(\mathvect{t}_{m})$ }
      \UnaryInfC{$P_{i}(\mathvect{t})$}
    \end{inlineprooftree}. 
  \end{center}
  
  The formulas above the line of a production are called the \emph{assumptions} of the production. 
  The formula under the line of a production is called the \emph{conclusion} of the production. 
  An \emph{inductive definition set} is a finite set of productions.
\end{definition}

The unary inductive predicates $\TeFsy$ and $\FsTsy$ given in the next example are important in this paper
because they will work as a counterexample for showing the failure of cut-elimination in a cyclic proof system. 

\begin{example}\label{ex:TeF_def}
  The productions for $\TeFsy$ and $\FsTsy$ are given as follows:
  \begin{center}
    $\begin{inlineprooftree}
      \AxiomC{}
      \UnaryInfC{$\TeF{\e}$}
    \end{inlineprooftree}$, 
    \qquad
    $\begin{inlineprooftree}
      \AxiomC{$\TeF{\nx{x}}$}
      \UnaryInfC{$\TeF{x}$}
    \end{inlineprooftree}$,
    \qquad
    $\begin{inlineprooftree}
      \AxiomC{}
      \UnaryInfC{$\FsT{\s}$}
    \end{inlineprooftree}$,
    \qquad
    $\begin{inlineprooftree}
      \AxiomC{$\FsT{x}$}
      \UnaryInfC{$\FsT{\nx{x}}$}
    \end{inlineprooftree}$, 
  \end{center}
  where $\s$ (start) and $\e$ (end) are constant symbols, $\nx{x}$ (``next of $x$'') is a unary function symbol. 

  Intuitively, $\TeF{t}$ (read ``to $\e$ from $t$'') means that, for some $m \ge 0$,
  ``the $m$-th next element of $t$ is $\e$'', 
  since $\TeF{t}$ holds for $t$ such that $\e = \nx^m t$ for some $m\ge 0$. 
  Also, $\FsT{t}$ (read ``from $\s$ to $t$'') means that, for some $m\ge 0$,
  ``the $m$-th next element from $\s$ is $t$'', 
  since $\FsT{t}$ holds for $t$ such that $t = \nx^m\s$ for some $m \ge 0$. 
  Hence $\TeF{\s}$ and $\FsT{\e}$ are semantically same,
  that is, they both mean $\e = \nx^m \s$ for some $m\ge 0$. 
\end{example}

The semantics of an inductive predicate is given by the standard least fixed point semantics,
namely, the least fixed point of a monotone operator constructed from its productions (see \cite{Brotherston2011}).
We skip giving its detailed definition because we do not use semantics in this paper.

\begin{definition}[Sequent]
  Let $\Gamma$ and $\Delta$ be finite sets of formulas in $\Base$. 
  A \emph{sequent} (denoted by $\Sequent$) of $\Base$ is a pair $\Gamma \fCenter \Delta$. 
  The first set $\Gamma$ is called the \emph{antecedent} of $\Gamma \fCenter \Delta$ and the second one $\Delta$ 
  is called the \emph{succedent} of $\Gamma \fCenter \Delta$. 
\end{definition}

We use usual abbreviations like $\Gamma,\varphi\fCenter \varphi',\Delta$ for
$\Gamma\cup\{\varphi\} \vdash \{\varphi'\}\cup\Delta$,
and $\Gamma[\theta]$ for $\{\varphi[\theta] \mid \varphi\in\Gamma\}$. 
We define $\mathFV{\Gamma}$ as the union of free variables of formulas in $\Gamma$.

%#!latexmk -c -gg -lualatex main.tex
\section{Cyclic proof system $\CBaseOm$ for the base system}\label{sec:proof-system}

In this section, we define a cyclic proof system $\CBaseOm$ for the base system $\Base$.
To define it, we first define an infinitary proof system $\BaseOm$ for $\Base$ in the subsection~\ref{subsec:InfSys}. 
After that, $\CBaseOm$ is defined in the subsection~\ref{subsec:CyclicSys}. 
%These systems are the same as $\CBaseID$ and $\BaseIDOm$ 
%defined in \cite{BrotherstonPhD, Brotherston2011} except for small details.

\subsection{Inference rules}
\label{subsec:inference_rules}

This section gives the common inference rules for both $\BaseOm$ and $\CBaseOm$. 
The inference rules except for rules of inductive predicates 
are given in Figure \ref{fig:inference-rules}.
The sequents above the line of a rule are called the \emph{assumptions} of the rule. 
The sequent under the line of a rule is called the \emph{conclusion} of the rule. 
The \emph{principal formula} of a rule is the distinguished formula in its conclusion.
The distinguished formulas ($\varphi$ in Figure \ref{fig:inference-rules}) of $\ruleCut$
are called the \emph{cut-formulas}. 

\begin{figure}[t]%
  \centering
  \begin{tabular}[tb]{l}
    \begin{tabular}{cc}
      \begin{minipage}{0.5\hsize}
        \begin{prooftree}
          \AxiomC{}
          \RightLabel{$\ruleAxiom$ ($\Gamma \cap \Delta \neq \emptyset$)}          
          \UnaryInfC{$\Gamma \fCenter \Delta$}
        \end{prooftree}
        \smallskip
      \end{minipage}
      &
      \begin{minipage}{0.5\hsize}
	    \begin{prooftree}
	      \AxiomC{$ \Gamma' \fCenter \Delta'$}
	      \RightLabel{$\ruleWeak$ ($\Gamma' \subseteq \Gamma$, $\Delta' \subseteq \Delta$)}
	      \UnaryInfC{$ \Gamma \fCenter \Delta $}
	    \end{prooftree}
	    \smallskip
      \end{minipage}
      \\
      \begin{minipage}{0.5\hsize}
        \begin{prooftree}
          \AxiomC{$\Gamma \fCenter \varphi, \Delta$}
          \AxiomC{$\Gamma, \varphi \fCenter \Delta$}
          \RightLabel{$\ruleCut$}
          \BinaryInfC{$\Gamma \fCenter \Delta$}
        \end{prooftree}
        \smallskip
      \end{minipage}
      &
      \begin{minipage}{0.5\hsize}
	    \begin{prooftree}
	      \AxiomC{$\Gamma \fCenter \Delta$}
	      \RightLabel{$\ruleSubst$}
	      \UnaryInfC{$\Gamma \left[ \theta \right] \fCenter \Delta \left[ \theta \right] $}
	    \end{prooftree}
	    \smallskip
      \end{minipage} 
    \end{tabular} \\

    \begin{tabular}{cc}
      \begin{minipage}{0.5\hsize}
        \begin{prooftree}
          \AxiomC{$\Gamma \left[x := u, y := t\right], t = u \fCenter \Delta \left[x := u, y := t\right] $}
          % \LeftLabel{$x$, $y\notin \mathVar{t}\cup\mathVar{u}$}
          \RightLabel{$\ruleEqLa$}
          \UnaryInfC{$ \Gamma \left[ x := t, y := u \right], t = u \fCenter \Delta \left[x := t, y := u\right] $}
        \end{prooftree}
        \smallskip
      \end{minipage}  
      &
      \begin{minipage}{0.5\hsize}
        \begin{prooftree}
          \AxiomC{\phantom{$ \Gamma  \fCenter t = t, \Delta $}}
          \RightLabel{$\ruleEqR$}
          \UnaryInfC{$ \Gamma  \fCenter t = t, \Delta $}
        \end{prooftree}
        \smallskip
      \end{minipage} 
    \end{tabular}
  \end{tabular}
  \caption{Inference rules except rules for inductive predicates} 
  \label{fig:inference-rules}
\end{figure}

We note that $\Gamma,\varphi,\varphi \vdash \Delta,\varphi',\varphi'$ is identified with
$\Gamma,\varphi \vdash \Delta,\varphi'$, 
since $\Gamma\cup\{\varphi,\varphi\} = \Gamma\cup\{\varphi\}$ and
$\Delta\cup\{\varphi',\varphi'\} = \Delta\cup\{\varphi'\}$.
Hence we do not have the contraction rule as an explicit inference rule. 

From a technical reason, we adopt $\ruleEqLa$ instead of $\ruleEqL$, which is given by:
\begin{center}
  \begin{inlineprooftree}
    \AxiomC{$\Gamma \left[x := u, y := t\right] \fCenter \Delta \left[x := u, y := t\right] $}
    \RightLabel{$\ruleEqL$}
    \UnaryInfC{$ \Gamma \left[ x := t, y := u \right], t = u \fCenter \Delta \left[x := t, y := u\right] $}
  \end{inlineprooftree}. 
\end{center}
These rules are derivable each other without adding extra $\ruleCut$:
$\ruleEqL$ is derivable by applying $\ruleWeak$ and $\ruleEqLa$, and
$\ruleEqLa$ is derivable by applying $\ruleEqL$.

We present the two inference rules for inductive predicates.  
First, for each production 
\begin{center}
  \begin{inlineprooftree}
    \AxiomC{ $ Q_1(\mathvect{u}_1\mleft( \mathvect{x} \mright)) \quad \cdots \quad Q_h(\mathvect{u}_h\mleft( \mathvect{x} \mright)) \quad P_{j_1}(\mathvect{t}_1  \mleft( \mathvect{x} \mright))\quad \cdots \quad P_{j_m}(\mathvect{t}_m  \mleft( \mathvect{x} \mright))$ }
    \UnaryInfC{$ P_i(\mathvect{t} \mleft( \mathvect{x} \mright))$}
  \end{inlineprooftree},
\end{center}
there is the inference rule 
\begin{center} 
  \begin{inlineprooftree}
    \AxiomC{ $\Gamma \fCenter Q_1(\mathvect{u}_1\mleft(\mathvect{u}\mright)), \Delta \,\cdots\, \Gamma \fCenter Q_h(\mathvect{u}_h\mleft(\mathvect{u}\mright)), \Delta$ \quad  $\Gamma \fCenter P_{j_1}(\mathvect{t}_1\mleft(\mathvect{u}\mright)), \Delta \,\cdots\, \Gamma \fCenter P_{j_{m}}(\mathvect{t}_m\mleft(\mathvect{u}\mright)), \Delta $}
    \RightLabel{\ruleUR{$P_{i}$}}
    \UnaryInfC{$\Gamma \fCenter P_i(\mathvect{t}\mleft(\mathvect{u}\mright)), \Delta$}
  \end{inlineprooftree}.
\end{center}

Next, we define the left introduction rule for the inductive predicate.
A \emph{case distinction} of $\Gamma , P_{i}(\mathvect{u}) \fCenter \Delta$ is defined as a sequent
\[
\Gamma, \mathvect{u} = \mathvect{t} \mleft( \mathvect{y} \mright), Q_1(\mathvect{u}_1\mleft( \mathvect{y} \mright)), \ldots , Q_h(\mathvect{u}_h\mleft( \mathvect{y} \mright)) , P_{j_1}(\mathvect{t}_1 \mleft( \mathvect{y} \mright)), \ldots , P_{j_m}(\mathvect{t}_m  \mleft( \mathvect{y} \mright)) \vdash \Delta,
\]
where $\mathvect{y}$ is a sequence of distinct variables of the same length as $\mathvect{x}$ and 
$y \not \in \mathFV{\Gamma \cup \Delta \cup \mathsetextension{ P_i(\mathvect{u})}} $ 
for all $y \in \mathvect{y}$, and there is a production
\begin{center}
  \begin{inlineprooftree}
    \AxiomC{$Q_1(\mathvect{u}_1\mleft( \mathvect{x} \mright)) \quad \ldots \quad Q_h(\mathvect{u}_h\mleft( \mathvect{x} \mright)) \quad P_{j_1}(\mathvect{t}_1 \mleft( \mathvect{x} \mright))\quad \ldots \quad P_{j_m}(\mathvect{t}_m \mleft( \mathvect{x} \mright))$ }
    \UnaryInfC{$ P_{i}(\mathvect{t} \mleft( \mathvect{x} \mright))$}
  \end{inlineprooftree}.
\end{center}
The inference rule $\ruleCase{P_{i}}$ is
\begin{center}
  \begin{inlineprooftree}
    \AxiomC{All case distinctions of $\Gamma , P_{i}(\mathvect{u}) \fCenter \Delta$}
    \RightLabel{$\ruleCase{P_{i}}$}
    \UnaryInfC{$\Gamma , P_{i}(\mathvect{u}) \fCenter \Delta $}
  \end{inlineprooftree}.
\end{center}
The formulas 
$P_{j_{1}}(\mathvect{t}_{1} \mleft( \mathvect{y} \mright)), \ldots, P_{j_{m}}(\mathvect{t}_{m} \mleft( \mathvect{y} \mright))$ 
in case distinctions are said to be \emph{case-descendants} of 
the principal formula $P_{i}(\mathvect{u})$. 

\begin{example}\label{ex:N_rule}
  The inference rules for the natural number predicate $\Nsy$ are
  \begin{center}
    \begin{inlineprooftree}
      \AxiomC{}
      \RightLabel{$\ruleURone{\Nrule}$}
      \UnaryInfC{$\Gamma \vdash \Delta,\N{\zero}$}
    \end{inlineprooftree}, 
    \qquad
    \begin{inlineprooftree}
      \AxiomC{$\Gamma \vdash \Delta,\N{t}$}
      \RightLabel{$\ruleURtwo{\Nrule}$}
      \UnaryInfC{$\Gamma \vdash \Delta,\N{\suc{t}}$}
    \end{inlineprooftree},
  \end{center}
  \begin{center}  
    \begin{inlineprooftree}
      \AxiomC{$\Gamma, t = \zero \vdash \Delta$}
      \AxiomC{$\Gamma, t = \suc{y}, \N{y} \vdash \Delta$}    
      \RightLabel{$\ruleCase{\Nrule}$}
      \BinaryInfC{$\Gamma, \N{t} \vdash \Delta$}
    \end{inlineprooftree}, 
  \end{center}
  where $y$ is a fresh variable. 
\end{example}

\begin{example}\label{ex:TeF_rule}
  The inference rules for the inductive predicates $\TeFsy$ and $\FsTsy$ given in Example~\ref{ex:TeF_def} are
  as follows:
  \begin{center}
    \begin{inlineprooftree}
      \AxiomC{}
      \RightLabel{$\ruleURone{\TeFrule}$}
      \UnaryInfC{$\Gamma \vdash \Delta,\TeF{\e}$}
    \end{inlineprooftree}, 
    \qquad
    \begin{inlineprooftree}
      \AxiomC{$\Gamma \vdash \Delta,\TeF{\nx{t}}$}
      \RightLabel{$\ruleURtwo{\TeFrule}$}
      \UnaryInfC{$\Gamma \vdash \Delta,\TeF{t}$}
    \end{inlineprooftree},
    \medskip
    \begin{inlineprooftree}
      \AxiomC{$\Gamma, t = \e \vdash \Delta$}
      \AxiomC{$\Gamma, t = y, \TeF{\nx{y}} \vdash \Delta$}    
      \rulenamelabel{Case $\TeFsy$}
      \BinaryInfC{$\Gamma, \TeF{t} \vdash \Delta$}
    \end{inlineprooftree},
    where $y$ is a fresh variable;
  \end{center}
  
  \begin{center}
    \begin{inlineprooftree}
      \AxiomC{}
      \RightLabel{$\ruleURone{\FsTrule}$}
      \UnaryInfC{$\Gamma \vdash \Delta,\FsT{\s}$}
    \end{inlineprooftree}, 
    \qquad
    \begin{inlineprooftree}
      \AxiomC{$\Gamma \vdash \Delta,\FsT{t}$}
      \RightLabel{$\ruleURtwo{\FsTrule}$}
      \UnaryInfC{$\Gamma \vdash \Delta,\FsT{\nx{t}}$}
    \end{inlineprooftree}, 
    \medskip
    \begin{inlineprooftree}
      \AxiomC{$\Gamma, t = \s \vdash \Delta$}
      \AxiomC{$\Gamma, t = \nx{y}, \FsT{y} \vdash \Delta$}    
      \rulenamelabel{Case $\FsTsy$}
      \BinaryInfC{$\Gamma, \FsT{t} \vdash \Delta$}
    \end{inlineprooftree},
    where $y$ is a fresh variable. 
  \end{center}
  
\end{example}

\subsection{Infinitary proof system $\BaseOm$}\label{subsec:InfSys}
In this subsection,
we define an infinitary proof system $\BaseOm$ for $\Base$.
The inference rules of $\BaseOm$ are the rules displayed in Figure~\ref{fig:inference-rules}
and the rules for inductive predicates given in the previous subsection. 

We write $\mleft\langle n_{1}, \dots, n_{k} \mright\rangle$ 
for the sequence of natural numbers $n_{1}, \dots, n_{k}$.
The length $|\sigma|$ of a sequent $\sigma$ is defined by the number of elements in $\sigma$. 
Let $\mathnatKC$ be the set of finite sequences of natural numbers.
We write $\sigma_{1}\sigma_{2}$ for the concatenation of $\sigma_{1}$ and $\sigma_{2}$ in $\mathnatKC$.
We abbreviate $\sigma \mleft\langle n \mright\rangle$ by $\sigma n$
for $\sigma\in\mathnatKC$ and $n\in\mathnat$.

Let $\mathRuleSet$ and $\mathSeqSet$ be the set of names of the inference rules and 
the set of sequents of $\BaseOm$, respectively.

\begin{definition}[Derivation tree]
  We define a \emph{derivation tree} to be a partial function
  $\mathdefparfunc{\DerivTree}{\mathnatKC}{\mathSeqSet\times\mleft(\mathRuleSet\cup\mathsetextension{\text{\rulename{Bud}}}\mright)}$
  satisfying the following conditions:
  \begin{enumerate}
  \item
    The domain $\Dom{\DerivTree}$ of $\DerivTree$ is prefixed-closed, namely,
    for $\sigma_{1}$, $\sigma_{2}\in\mathnatKC$,
    $\sigma_{1}\sigma_{2} \in \Dom{\DerivTree}$
    implies
	$\sigma_{1}\in\Dom{\DerivTree}$.
  \item
    If $\sigma n\in\Dom{\DerivTree}$ for $\sigma\in\mathnatKC$ and $n\in\mathnat$,  
	then $\sigma m\in\mathof{\mathdom}{\mathdertree{D}}$ for any $m\leq n$.
  \item
    For each $\sigma \in \Dom{\DerivTree}$, 
    we write $(\DSeq{\DerivTree}{\sigma},\DRule{\DerivTree}{\sigma})$ for $\mathof{\DerivTree}{\sigma}$.
    Then the following hold. 
    \begin{enumerate}
    \item
      If $\DRule{\DerivTree}{\sigma} = \text{\ruleBud}$, then $\sigma 0\notin \Dom{\DerivTree}$. 
    \item
      If $\DRule{\DerivTree}{\sigma} \neq \text{\rulename{Bud}}$, 
      $\sigma\,(n+1)\not\in\Dom{\DerivTree}$, and $\sigma 0,\ldots,\sigma n\in\Dom{\DerivTree}$, then
      the following is a rule instance of the rule $\DRule{\DerivTree}{\sigma}$: 
	  \begin{prooftree}
	    \AxiomC{$\DSeq{\DerivTree}{\sigma 0}$}
	    \AxiomC{$\cdots$}
	    \AxiomC{$\DSeq{\DerivTree}{\sigma n}$}
        \RightLabel{$\DRule{\DerivTree}{\sigma}$}
	    \TrinaryInfC{$\DSeq{\DerivTree}{\sigma}$}
	  \end{prooftree}
    \end{enumerate}
  \end{enumerate}
\end{definition}

An element in the domain of a derivation tree is called a \emph{node}.
The empty sequence as a node is called the \emph{root}.
The node $\sigma$ is called a \emph{bud} if $\DRule{\DerivTree}{\sigma}$ is $\ruleBud$.
We write $\Bud{\DerivTree}$ for the set of buds in $\DerivTree$. 
The node which is not a bud is called an \emph{inner node}.
A derivation tree is called \emph{infinite} if the domain of the derivation tree is infinite. 

We sometimes identify a node $\sigma$ with the sequent $\DSeq{\DerivTree}{\sigma}$.

\begin{definition}[Path] 
  We define a \emph{path} in a derivation tree $\DerivTree$ to be a (possibly infinite) sequence 
  $\mleft( \sigma_{i} \mright)_{0\leq i < \alpha}$ of nodes in $\Dom{\DerivTree}$
  such that $\sigma_{i+1}=\sigma_{i}n$ for some $n\in\mathnat$
  and $\alpha\in \mathnatpos \cup \mathsetextension{\omega}$, where 
  $\mathnatpos$ is the set of positive natural numbers and $\omega$ is the least infinite ordinal.
  A finite path $\sigma_{0}, \sigma_{1}, \dots, \sigma_{n}$ is called
  \emph{a path from $\sigma_{0}$ to $\sigma_{n}$}.
  The \emph{length of a finite path} $\mleft( \sigma_{i} \mright)_{0\leq i < \alpha}$ 
  is defined as $\alpha$.
  We define \emph{the height of a node} as the length of the path from the root to the node.
\end{definition}

We sometimes write $\mleft(\Gamma_{i} \fCenter \Delta_{i} \mright)_{0\leq i <\alpha}$
for the path $\mleft(\sigma_{i}\mright)_{0\leq i <\alpha}$ in a derivation tree $\mathdertree{D}$
if $\DSeq{\DerivTree}{\sigma_{i}} = \Gamma_{i} \fCenter \Delta_{i}$.

\begin{definition}[Trace] 
  For a path $\mleft(\Gamma_{i}\fCenter \Delta_{i}\mright)_{0\leq i <\alpha}$
  in a derivation tree $\mathdertree{D}$,
  we define a \emph{trace following}  $\mleft(\Gamma_{i}\fCenter\Delta_{i}\mright)_{0\leq i <\alpha}$ 
  to be a sequence of formulas $\mleft(\tau_{i} \mright)_{0\leq i <\alpha}$ 
  such that the following hold:
  \begin{enumerate}
  \item
    $\tau_{i}$ is an inductive predicate in $\Gamma_{i}$. 
  \item
    If $\Gamma_{i} \fCenter \Delta_{i}$ is the conclusion of $\ruleSubst$ with $\theta$, 
	then $\tau_{i}$ is $\tau_{i+1}\mleft[\theta\mright]$.
  \item
    If $\Gamma_{i}\fCenter\Delta_{i}$ is the conclusion of $\ruleEqLa$
	with the principal formula $t=u$ and
	$\tau_{i}$ is $\varphi\mleft[x:=t, y:=u\mright]$,
    then $\tau_{i+1}$ is $\varphi\mleft[x:=u, y:=t\mright]$.
  \item
    If $\Gamma_{i}\fCenter \Delta_{i}$ is  the conclusion of $\ruleCase{P_{i}}$, 
	then either
	\begin{itemize}
	\item
      $\tau_{i}$ is the principal formula of the rule and $\tau_{i+1}$ 
	  is a case-descendant of $\tau_{i}$, or
	\item
      $\tau_{i+1}$ is the same as $\tau_{i}$. 
	\end{itemize}	
	In the former case, $\tau_{i}$ is said to be a \emph{progress point} of the trace.
  \item
    If $\Gamma_{i} \fCenter \Delta_{i}$ is the conclusion of any other rules and $i+1<\alpha$, 
	then $\tau_{i+1}$ is $\tau_{i}$.
  \end{enumerate}
\end{definition}

\begin{definition}[Global trace condition]
  If a trace has infinitely many progress points,
  we call the trace an \emph{infinitely progressing trace}.
  If there exists an infinitely progressing trace following a tail of the path
  $\mleft( \Gamma_{i} \fCenter \Delta_{i} \mright)_{i \geq k}$ with some $k \geq 0$
  for every infinite path $\mleft( \Gamma_{i}\fCenter\Delta_{i} \mright)_{i \geq 0}$ in a derivation tree,
  we say the derivation tree satisfies the \emph{global trace condition}.
\end{definition}

\begin{definition}[$\BaseOm$ pre-proof]
  A (possibly infinite) derivation tree $\DerivTree$ without buds is called a $\BaseOm$ \emph{pre-proof}.
  The sequent $\DSeq{\DerivTree}{\mleft\langle\mright\rangle}$ at the root node 
  is called the conclusion of $\DerivTree$.
\end{definition}

\begin{definition}[$\BaseOm$ proof]
  A $\BaseOm$ pre-proof that satisfies the global trace condition is called a $\BaseOm$ \emph{proof}. 
\end{definition}

The global trace condition was originally introduced as a sufficient condition for the soundness of
$\LKIDOm$ with respect to the standard models in Brotherston's paper~\cite{BrotherstonPhD, Brotherston2011}.
It also ensures the soundness of $\BaseOm$, since 
a $\BaseOm$ proof can be transformed to a $\LKIDOm$ proof by replacing $\ruleEqLa$ by $\ruleEqL$. 
Brotherston also showed the cut-free completeness of $\LKIDOm$ for the standard models. 
The cut-free completeness of $\BaseOm$ follows from this result: 
a cut-free $\LKIDOm$ proof of a sequent in $\BaseOm$ can contain only the rules 
$\ruleAxiom$, $\ruleWeak$, $\ruleSubst$, $\ruleEqL$, $\ruleEqR$, and the rules for inductive predicates.
Hence the cut-free proof can be transformed to a cut-free $\BaseOm$ proof
by replacing $\ruleEqL$ by $\ruleEqLa$ with $\ruleWeak$. 

\subsection{Cyclic proof system $\CBaseOm$} \label{subsec:CyclicSys}

In this section, we introduce a cyclic proof system $\CBaseOm$. 

\begin{definition}[Companion]
  For a finite derivation tree $\DerivTree$,
  we define the \emph{companion} for a bud $\sigma_{\textit{bud}}$ as an inner node $\sigma$ in $\DerivTree$ with 
  $\DSeq{\DerivTree}{\sigma}=\DSeq{\DerivTree}{\sigma_{bud}}$.
\end{definition}

\begin{definition}[$\CBaseOm$ pre-proof]
  We define a $\CBaseOm$ \emph{pre-proof}
  to be a pair $\mleft(\DerivTree, \BudCompanion\mright)$
  such that $\DerivTree$ is a finite derivation tree and 
  $\BudCompanion$ is a function mapping each bud to its companion.
  The sequent at the root node of $\DerivTree$ is called the conclusion of the proof.
\end{definition}

\begin{definition}[Tree-unfolding]
  Let $\ProofTree$ be a $\CBaseOm$ pre-proof $\mleft(\DerivTree,\BudCompanion\mright)$. 
  A \emph{tree-unfolding} $\TreeUnfold{\ProofTree}$ of $\ProofTree$ is recursively defined by
  \[ 
%  \mathof{\mathtreeunfolding{\DerivTree,\BudCompanion}}{\sigma}
  \mathUnaryOf{\TreeUnfold{\ProofTree}}{\sigma}
  =
  \begin{cases}
    \mathUnaryOf{\mathdertree{D}}{\sigma}, 
    & \text{if $\sigma\in\Dom{\DerivTree}\setminus\Bud{\DerivTree}$,}
    \\ %
    \mathUnaryOf{\TreeUnfold{\ProofTree}}{\sigma_{3}\sigma_{2}}, 
    & \text{if $\sigma\notin\Dom{\DerivTree}\setminus\Bud{\DerivTree}$ with $\sigma=\sigma_{1}\sigma_{2}$, $\sigma_{1}\in \Bud{\DerivTree}$ and $\sigma_{3}=\mathUnaryOf{\BudCompanion}{\sigma_{1}}$,}
  \end{cases}
  \]
\end{definition}

Note that a tree-unfolding is a $\BaseOm$ pre-proof. 

\begin{definition}[$\CBaseOm$ proof]
  A $\CBaseOm$ pre-proof $\ProofTree$ of a sequent $\Sequent$ is called a $\CBaseOm$ \emph{proof} of $\Sequent$
  if its tree-unfolding $\TreeUnfold{\ProofTree}$ satisfies the global trace condition.  
  A \emph{cut-free} $\CBaseOm$ proof is a $\CBaseOm$ proof that does not contain $\ruleCut$.
  A sequent $\Sequent$ is said to be \emph{(cut-free) provable} in $\CBaseOm$
  if a (cut-free) $\CBaseOm$ proof of $\Sequent$ exists. 
\end{definition}

% \begin{example}%
%  \label{ex:proofAddone}
%  The derivation tree given in Figure \ref{fig:counter-ex-D1}
%  is the proof of $\mathindAddone{x_1}{sy_1}{z_1} \fCenter \mathindAddone{sx_1}{y_1}{z_1}$ in \CLKID,  
%  where ($\star$) indicates the pairing of a companion with a bud and
%  the underlined formulas are the infinitely progressing trace 
%  for the infinite path (some applying rules and some labels are omitted for limited space).
%  \input{fig_counter-ex-D1}
% \end{example}

A $\CBaseOm$ pre-proof in which each companion is an ancestor of the corresponding bud is called \emph{cycle-normal}.
The following proposition says that $\CBaseOm$ satisfies the cycle-normalization property. 

\begin{prop}\label{prop:cyclenormalization}
  For a $\CBaseOm$ pre-proof $\ProofTree$, we have
  a $\CBaseOm$ cycle-normal pre-proof $\ProofTree'$ such that
  $\TreeUnfold{\ProofTree} = \TreeUnfold{\ProofTree'}$. 
\end{prop}

This property was already shown by Brotherston in a general setting that includes $\CBaseOm$~\cite{BrotherstonPhD}.
Besides it, a shorter proof for $\CLKID$ was given in \cite{Masuoka2021}.
It can be applied straightforwardly to the current setting. 
(We give the proof in the appendix for the reviewer's convenience.)

%#!latexmk -c -gg -lualatex
 \section{A counterexample with only unary inductive predicates to cut-elimination in \CBaseOm}
 \label{sec:main}
 In this section,  we prove the following theorem, which is the main theorem.
 Let $\mathindToEFromsy$ and $\mathindFromSTosy$ be the inductive predicates defined in \cref{ex:TeF_def}.

 \begin{theorem} 
  \label{thm:main}
  The following statements hold:
  \begin{enumerate}
   \item $\mathindToEFrom{\mathstart} \fCenter \mathindFromSTo{\mathend}$ is provable in \CBaseOm.
	 \label{item:thm-main-provable}
   \item $\mathindToEFrom{\mathstart} \fCenter \mathindFromSTo{\mathend}$ is not cut-free provable 
	 in \CBaseOm.
	 \label{item:thm-main-not-cut-free-provable}
  \end{enumerate}
 \end{theorem}

 This theorem means that
 $\mathindToEFrom{\mathstart} \fCenter \mathindFromSTo{\mathend}$
 is a counterexample with only unary inductive predicates to cut-elimination in \CBaseOm.

 A \CBaseOm\ proof of $\mathindToEFrom{\mathstart} \fCenter \mathindFromSTo{\mathend}$ is
 the derivation tree given in \cref{fig:counter-ex},
 where ($\dagger$) indicates the pairing of the companion with the bud and
 the underlined formulas denotes the infinitely progressing trace for the tails of the infinite path
 (some applying rules and some labels of rules are omitted for limited space).
 Thus, \cref{thm:main} \cref{item:thm-main-provable} is correct.
  
  \begin{figure}[tb]
   \centering
   \begin{prooftree}
    \tiny
    \AxiomC{}
    \UnaryInfC{$\fCenter \mathindFromSTo{\mathstart}$} 
    \UnaryInfC{${\mathstart = \mathend} \fCenter \mathindFromSTo{\mathend}$} 
    %(1)

    \AxiomC{}
    \UnaryInfC{$ \fCenter \mathindFromSTo{\mathstart}$} 
    \UnaryInfC{$ \fCenter \mathindFromSTo{\mathnext \mathstart}$} 
    \UnaryInfC{$\mathstart = x \fCenter \mathindFromSTo{\mathnext x}$} 
    %(2)

    \AxiomC{}
    \UnaryInfC{$\mathindFromSTo{\mathend} \fCenter \mathindFromSTo{\mathend}$} 
    \UnaryInfC{${\mathnext x = \mathend}, \mathindFromSTo{\mathnext x} \fCenter \mathindFromSTo{\mathend}$} 
    %(3)

    \AxiomC{}
    \UnaryInfC{$ \mathindFromSTo{\mathnext x} \fCenter \mathindFromSTo{\mathnext x}$}
    \UnaryInfC{$ \mathindFromSTo{\mathnext x} \fCenter \mathindFromSTo{\mathnext \mathnext x}$}
    %(4)

    \AxiomC{($\dagger$)\, $\underline{\mathindToEFrom{\mathnext x}}, \mathindFromSTo{\mathnext x} \fCenter \mathindFromSTo{\mathend}$}
    \UnaryInfC{$\underline{\mathindToEFrom{\mathnext \mathnext x}}, \mathindFromSTo{\mathnext \mathnext x} \fCenter \mathindFromSTo{\mathend}$}
    %(4)

    \rulenamelabel{Cut} %(4)
    \BinaryInfC{$\underline{\mathindToEFrom{\mathnext \mathnext x}}, \mathindFromSTo{\mathnext x} \fCenter \mathindFromSTo{\mathend}$}
    \UnaryInfC{${\mathnext x = y}, \underline{\mathindToEFrom{\mathnext y}}, \mathindFromSTo{\mathnext x} \fCenter \mathindFromSTo{\mathend}$} 
    %(3)

    \rulenamelabel{Case $\mathindToEFromsy$} %(3)
    \BinaryInfC{($\dagger$)\, $\underline{\mathindToEFrom{\mathnext x}}, \mathindFromSTo{\mathnext x} \fCenter \mathindFromSTo{\mathend}$} 
    \UnaryInfC{$\mathstart = x, \mathindToEFrom{\mathnext x}, \mathindFromSTo{\mathnext x} \fCenter \mathindFromSTo{\mathend}$} 
    %(2)

    \rulenamelabel{Cut} %(2)
    \BinaryInfC{$\mathstart = x, \mathindToEFrom{\mathnext x} \fCenter \mathindFromSTo{\mathend}$} 
    %(1)
    
    \rulenamelabel{Case $\mathindToEFromsy$} %(1)
    \BinaryInfC{$\mathindToEFrom{\mathstart} \fCenter \mathindFromSTo{\mathend}$}
   \end{prooftree}
   
   \caption{The \CBaseOm\ proof of $\mathindToEFrom{\mathstart} \fCenter \mathindFromSTo{\mathend}$} 
   \label{fig:counter-ex}
  \end{figure}

  In this section, 
  we henceforth prove \cref{thm:main} \cref{item:thm-main-not-cut-free-provable}.

  \subsection{The outline of the proof of \cref{thm:main} \cref{item:thm-main-not-cut-free-provable}}
  Before proving the theorem, 
  we outline our proof of \cref{thm:main} \cref{item:thm-main-not-cut-free-provable}.

 Assume there exists a cut-free \CBaseOm\ proof of
 $\mathindToEFrom{\mathstart} \fCenter \mathindFromSTo{\mathend}$.
 By the cycle-normalization property of \CBaseOm, there exists a cut-free cycle-normal \CBaseOm\ proof of 
 $\mathindToEFrom{\mathstart} \fCenter \mathindFromSTo{\mathend}$.
 Let $\mathproofcf$ be the \CBaseOm\ proof. 

 The key concepts for the proof are a \emph{root-like sequent}, a \emph{switching point}, and 
 an \emph{unfinished path}.
 To define these concepts,
 we define the relation $\mathequivrel{\Gamma}$ for a finite set of formulas $\Gamma$
 to be the smallest congruence relation on terms containing $t_{1} = t_{2} \in \Gamma$ (\cref{def:equivrel})
 and the \emph{index of $\mathindToEFrom{t}$ in a sequent $\Gamma\fCenter\Delta$} (\cref{def:index}).
 The index of $\mathindToEFrom{t}$ is $m-n$
 if there uniquely exists $m-n$ such that $n$, $m\in\mathnat$,
 and $\mathnext^{n}t \mathequivrel{\Gamma} \mathnext^{m} \mathstart$.
 The index of $\mathindToEFrom{t}$ is $\bot$
 if $\mathnext^{m}t \not\mathequivrel{\Gamma} \mathnext^{n} \mathstart$ for any $m$, $n\in\mathnat$.
 The index of $\mathindToEFrom{t}$ may be undefined,
 but the index is always defined in a special sequent, called a \emph{root-like sequent} 
 (\cref{def:invariant}).
 A \emph{switching point} is defined as a node 
 that is the conclusion of \rulename{Case $\mathindToEFromsy$} with
 the principal formula whose index is $\bot$ (\cref{def:switching_point}).
 An \emph{unfinished path} is defined as
 a path $\mleft( \Gamma_{i}\fCenter\Delta_{i} \mright)_{0\leq i <\alpha}$ 
 of $\mathtreeunfoldingnb{\mathproofcf}$ such that
 $\Gamma_{0}\fCenter\Delta_{0}$ is a root-like sequent and 
 $\Gamma_{i}\fCenter\Delta_{i}$ is a switching point
 if $\Gamma_{i+1}\fCenter\Delta_{i+1}$ is the left assumption of $\Gamma_{i}\fCenter\Delta_{i}$ 
 (\cref{def:bad-path}).
 Then, the following statements hold:
 \begin{enumerate}
  \item The root is a root-like sequent; \label{outlinelemma:root}
  \item Every sequent in an unfinished path is a root-like sequent (\cref{lemma:invariant});
	\label{outlinelemma:index-path}
  \item There exists a switching point on an infinite unfinished path (\cref{lemma:key_lemma}); and
	\label{outlinelemma:key}
  \item The rightmost path from a root-like sequent is infinite (\cref{lemma:rightmost}).
	\label{outlinelemma:rightmost-path}
 \end{enumerate}  

  At last, 
  we show there exist infinite nodes in the derivation tree $\mathdertreecf$.
  Because of \cref{outlinelemma:root} and \cref{outlinelemma:rightmost-path}, 
  the rightmost path from the root is an infinite unfinished path. 
  By \cref{outlinelemma:key}, there exists a switching point on the path.
  Let $\tilde{\sigma}_{0}$ be the node of the smallest height among such switching points.
  Let $\alpha_{0}$ be the left assumption of $\tilde{\sigma}_{0}$.
  By \cref{outlinelemma:index-path}, the sequent of $\alpha_{0}$ is a root-like sequent.
  By \cref{outlinelemma:rightmost-path}, the rightmost path from $\alpha_{0}$ is infinite.
  Therefore, there exists a bud $\mu_{0}$ in the rightmost path from $\alpha_{0}$.
  By \cref{outlinelemma:key} and the definition of $\tilde{\sigma}_{0}$, 
  there exists a switching point between $\alpha_{0}$ and $\mu_{0}$ .
  Let $\tilde{\sigma}_{1}$ be the node of the smallest height among such switching points.
  The nodes $\tilde{\sigma}_{0}$ and $\tilde{\sigma}_{1}$ are distinct by their definitions.
  By repeating this process as in \cref{fig:idea}, 
  we get a set of infinite nodes $\mathsetintension{\tilde{\sigma}_{i}}{i\in\mathnat}$.
  It is a contradiction since the set of nodes of $\mathdertreecf$ is finite.
  
  \begin{figure}[tbhp]
 \begin{prooftree}
	   \AxiomC{$\begin{aligned}&\vdots\\&\alpha_{2}\tikzmark{node-a2}\end{aligned}$}%(4)
	   \AxiomC{$\begin{aligned}&\mu_{1}\tikzmark{node-b1}\\&\vdots\end{aligned}$}
  \rulenamelabel{Case $\mathindToEFromsy$\tikzmark{node-M2}}
  \BinaryInfC{$\tilde{\sigma}_{2}$\tikzmark{node-c2}}%(4)
  \noLine
	      \UnaryInfC{$\begin{aligned}&\vdots\\&\alpha_{1}\tikzmark{node-a1}\end{aligned}$}%(3)
  
	   \AxiomC{$\begin{aligned}&\mu_{0}\tikzmark{node-b0}\\&\vdots\end{aligned}$} %(3)
  
  %(3)
  \rulenamelabel{Case $\mathindToEFromsy$\tikzmark{node-M1}}
  \BinaryInfC{$\tilde{\sigma}_{1}$\tikzmark{node-c1}}
  \noLine
	      \UnaryInfC{$\begin{aligned}&\vdots\\&\alpha_{0}\tikzmark{node-a0}\end{aligned}$}%(2)
	   \AxiomC{$\begin{aligned}&\mu\tikzmark{node-b}\\&\vdots\end{aligned}$}  %(2)

  \rulenamelabel{Case $\mathindToEFromsy$\tikzmark{node-Mr}}
	       \BinaryInfC{$\begin{aligned}&\tilde{\sigma}_{0}\tikzmark{node-c0}\\&\vdots\end{aligned}$}
  \noLine
  \UnaryInfC{$\mathindToEFrom{\mathstart} \fCenter$ \tikzmark{node-root} $\mathindFromSTo{\mathend}$}
 \end{prooftree}
 \begin{tikzpicture}[remember picture, overlay, thick, relative, auto, line width=0.3pt]
  \coordinate (b) at ({pic cs:node-b});
  \coordinate (root) at ({pic cs:node-root});
  \coordinate (M) at ({pic cs:node-Mr});
  \coordinate (c0) at ({pic cs:node-c0});
  \coordinate (b0) at ({pic cs:node-b0});
  \coordinate (a0) at ({pic cs:node-a0});
  \coordinate (M1) at ({pic cs:node-M1});
  \coordinate (c1) at ({pic cs:node-c1});
  \coordinate (b1) at ({pic cs:node-b1});
  \coordinate (a1) at ({pic cs:node-a1});
  \coordinate (M2) at ({pic cs:node-M2});
  \coordinate (c2) at ({pic cs:node-c2});
  
  \draw[->]($(b)$)..controls ($(M)+(5.0em, 0)$) .. ($(c0)!0.5!(root)$);

  \draw[->] ($(b0)$)..controls ($(M1)+(5.0em, 0)$) .. ($(c1)!0.5!(a0)$);

  \draw[->] ($(b1)$) ..controls ($(M2)+(5.0em, 0)$) .. ($(c2)!0.5!(a1)$);
 \end{tikzpicture}
 \caption{Construction of $\mleft(\tilde{\sigma}_{i}\mright)_{i\in\mathnat}$} 
 \label{fig:idea}
\end{figure}

  \subsection{The proof of \cref{thm:main} \cref{item:thm-main-not-cut-free-provable}}
  We show \cref{thm:main} \cref{item:thm-main-not-cut-free-provable}.
  Assume there exists a cut-free \CBaseOm\ proof of 
  $\mathindToEFrom{\mathstart} \fCenter \mathindFromSTo{\mathend}$
  for contradiction.
  By the cycle-normalization property of \CBaseOm, there exists a cut-free cycle-normal \CBaseOm\ proof of 
  $\mathindToEFrom{\mathstart} \fCenter \mathindFromSTo{\mathend}$.
  \emph{We write $\mleft(\mathdertreecf, \mathcompanioncf\mright)$ 
  for a cut-free cycle-normal \CBaseOm\ proof of $\mathindToEFrom{\mathstart} \fCenter \mathindFromSTo{\mathend}$.}

 \begin{rem}%
  \label{rem:cut-freeness}
  Let $\Gamma \fCenter \Delta$ be a sequent in $\mathproofcf$.
  By induction on the height of sequents in $\mathdertreecf$, we can easily show the following statements:
  
   \begin{enumerate}%
    \item $\Gamma$ consists of only atomic formulas with $=$, $\mathindToEFromsy$.
    \item $\Delta$ consists of only atomic formulas with $\mathindFromSTosy$.
    \item A term in $\Gamma$ and $\Delta$ is of the form $\mathnext^{n}\mathstart$,
	  $\mathnext^{n}\mathend$, or $\mathnext^{n}x$ with some variable $x$.
    \item The possible rules in $\mathproofcf$ are
	  \rulename{Weak}, \rulename{Subst}, \ruleEqLa, 
	  \rulename{$\mathindFromSTosy$ R${}_{\text{1}}$}, \rulename{$\mathindFromSTosy$ R${}_{\text{2}}$},
	  and \rulename{Case $\mathindToEFromsy$}
	  (See \cref{ex:TeF_rule}).
   \end{enumerate}
 \end{rem}

 Without loss of generality, we assume terms in this section are of the form
 $\mathnext^{n}\mathstart$, $\mathnext^{n}\mathend$, or $\mathnext^{n}x$
 with some variable $x$.
 
 % We define the equality $\mathequivrel{\Gamma}$ in a sequent $\Gamma\fCenter\Delta$.

 \begin{definition}[$\mathequivrel{\Gamma}$]%
  \label{def:equivrel}
  For a set of formulas $\Gamma$,
  we define the relation $\mathequivrel{\Gamma}$ to be the smallest congruence relation on terms
  which satisfies the condition that $t_{1} = t_{2} \in \Gamma$ implies $t_{1} \mathequivrel{\Gamma} t_{2}$.
 \end{definition}
 
 Intuitively, $\mathequivrel{\Gamma}$ represents the equal in any models of $\Gamma$.

 \begin{definition}[$\mathdeprel\Gamma$]%
  \label{def:deprel}
  For a set of formulas $\Gamma$ and terms $t_{1}$, $t_{2}$,
  we define $t_{1} \mathdeprel{\Gamma} t_{2}$
  by $\mathnext^{n} t_{1} \mathequivrel{\Gamma} \mathnext^{m} t_{2}$ for some  $n, m \in \mathnat$.
 \end{definition}

 Note that $\mathdeprel{\Gamma}$ is a congruence relation.

 \begin{definition}[Index]%
  \label{def:index}
  For a finite set $\Gamma$ and $\mathindToEFrom{t}\in\Gamma$,
  we define \emph{the index of $\mathindToEFrom{t}$ in $\Gamma$} as follows:
  \begin{enumerate}
   \item If $t \not\mathdeprel{\Gamma} \mathstart$,
	 then the index of $\mathindToEFrom{t}$ in $\Gamma$ is $\bot$, and
   \item if there uniquely exists $m-n$ such that $n$, $m\in\mathnat$,
	 and $\mathnext^{n}t \mathequivrel{\Gamma} \mathnext^{m} \mathstart$,
	 then the index of $\mathindToEFrom{t}$ in $\Gamma$ is $m-n$
	 (namely the uniqueness means that $\mathnext^{n'}t \mathequivrel{\Gamma} \mathnext^{m'} \mathstart$
	 for $n$, $m\in\mathnat$ implies $m-n=m'-n'$).
  \end{enumerate}
 \end{definition}

 Note that
 if there exists $n_{0}$, $m_{0}$, $n_{1}$, $m_{1}\in \mathnat$ 
 such that $\mathnext^{n_{0}}t \mathequivrel{\Gamma} \mathnext^{m_{0}} \mathstart$,
 $\mathnext^{n_{1}}t \mathequivrel{\Gamma} \mathnext^{m_{1}} \mathstart$ and $m_{0}-n_{0}\neq m_{1}-n_{1}$,
 then the index of $\mathindToEFrom{t}$ in $\Gamma$ is undefined.
 % Intuitively, It means 
 % $t^{\mathmodel{M}}=\mathof{{\mathnext^{\mathmodel{M}}}^{n}}{\mathstart^{\mathmodel{M}}}$ 
 % in a model $\mathmodel{M}$ of $\Gamma$ where $\mathnext^{\mathmodel{M}}$ is injective
 % that the index of $\mathindToEFrom{t}$ in $\Gamma$ is $n$.

 \begin{definition}[Root-like sequent]%
  \label{def:invariant}%
  The sequent $\Gamma \fCenter \Delta$ is said to be
  a \emph{root-like sequent} if the following conditions hold:
  \begin{enumerate}
   \item $\mathstart\not\mathdeprel{\Gamma}\mathend$,
	 \label{item:def-invariant_start_and_end}
   \item $t\not\mathdeprel{\Gamma}\mathstart$ for any $\mathindFromSTo{t}\in\Delta$, and
	 \label{item:def-invariant_existence_of_premise}
   \item if $\mathnext^{n}\mathstart\mathequivrel{\Gamma}\mathnext^{m}\mathstart$, then $n = m$.
	 \label{item:def-invariant_defined_index}
  \end{enumerate}
 \end{definition}

 %By the conditions \cref{item:def-invariant_start_and_end} and \cref{item:def-invariant_existence_of_premise},
 A root-like sequent does not occur as a conclusion of \rulename{$\mathindFromSTosy$ R${}_{\text{1}}$}
 by the first and second conditions.
 The third condition guarantees the existence of an index, as shown in the following lemma. 
 % We will use \cref{item:def-invariant_existence_of_premise} to calculate an index 
 % in \cref{lemma:index} \cref{item:lemma-index_neg} and an infinite sequence 
 % in \cref{lemma:rightmost}.

 \begin{lem}
  If $\Gamma\fCenter\Delta$ is a root-like sequent, 
  the index of any $\mathindToEFrom{t}$ in $\Gamma$ is defined.
 \end{lem}

 \begin{proof}
  Let $\mathindToEFrom{t}\in\Gamma$. 
  If $t\not\mathdeprel{\Gamma}\mathstart$, then the index is $\bot$.

  Assume $t\mathdeprel{\Gamma}\mathstart$.
  By \cref{def:deprel},
  there exist $n_{0}$ and $m_{0}$
  such that $\mathnext^{n_{0}} t \mathequivrel{\Gamma} \mathnext^{m_{0}} \mathstart$.
  To show the uniqueness,
  assume $\mathnext^{n_{1}} t \mathequivrel{\Gamma} \mathnext^{m_{1}} \mathstart$ for $n_{1}$ and $m_{1}$ .
  Since $\mathnext^{n_{0}+n_{1}} t \mathequivrel{\Gamma} \mathnext^{m_{0}+n_{1}} \mathstart$
  and $\mathnext^{n_{1}+n_{0}} t \mathequivrel{\Gamma} \mathnext^{m_{1}+n_{0}} \mathstart$,
  we have $\mathnext^{m_{0}+n_{1}} \mathstart \mathequivrel{\Gamma} \mathnext^{m_{1}+n_{0}} \mathstart$.
  From \cref{item:def-invariant_defined_index} of \nolinebreak \cref{def:invariant}, 
  $m_{0}+n_{1}=m_{1}+n_{0}$.
  Thus, $m_{0}-n_{0}=m_{1}-n_{1}$.
 \end{proof}

 \begin{definition}[Switching point]%
  \label{def:switching_point}
  A node $\sigma$ in a derivation tree is called a \emph{switching point}
  if the rule with the conclusion $\sigma$ is \rulename{Case $\mathindToEFromsy$} and
  the index of the principal formula for the rule in the conclusion is $\bot$.
 \end{definition}

  We call the assumption of \rulename{Case $\mathindToEFromsy$} whose form is
  $\Gamma, {t = x}, \mathindToEFrom{\mathnext x} \fCenter \Delta$ 
  \emph{the right assumption of} the rule.
  The other assumption is called \emph{the left assumption of} the rule. 

 \begin{definition}[Unfinished path] 
  \label{def:bad-path}
  A path $\mleft( \Gamma_{i} \fCenter \Delta_{i} \mright)_{0\leq i <\alpha}$ 
  in $\mathtreeunfoldingnb{\mathproofcf}$
  with some $\alpha \in \mathnat \cup \mathsetextension{\omega}$
  is said to be an \emph{unfinished path}
  if the following conditions hold:
  \begin{enumerate}
   \item $\Gamma_{0}\fCenter\Delta_{0}$ is a root-like sequent, and
   \item if the rule for $\Gamma_{i} \fCenter \Delta_{i}$ is $\text{\rulename{Case $\mathindToEFromsy$}}$
	 and $\Gamma_{i+1}\fCenter\Delta_{i+1}$ is the left assumption of the rule,
	 then $\Gamma_{i}\fCenter\Delta_{i}$ is a switching point.
  \end{enumerate}
 \end{definition}

 \begin{lem}%
  \label{lemma:invariant}
  Every sequent in an unfinished path is a root-like sequent.
 \end{lem}

 \begin{proof}[Sketch of proof]
  Let $\mleft( \Gamma_{i} \fCenter \Delta_{i} \mright)_{0\leq i <\alpha}$ be an unfinished path. 
  We prove the statement by the induction on $i$.
  
  For $i=0$, $\Gamma_{0} \fCenter \Delta_{0}$ is a root-like sequent by \cref{def:bad-path}.

  For $i>0$, we can prove the statement 
  by considering cases according to the rule with the conclusion $\Gamma_{i-1}\fCenter\Delta_{i-1}$.
  For more details, see \cref{appendix:lemma-invariant}.
 \end{proof}

 \begin{lem}%
  \label{lemma:index}
  For an unfinished path $\mleft(\Gamma_{i}\fCenter\Delta_{i}\mright)_{0\leq i< \alpha}$ and
  a trace $\mleft( \tau_{k} \mright)_{k \geq 0}$ following 
  $\mleft(\Gamma_{i}\fCenter\Delta_{i}\mright)_{i\geq p}$,
  if $d_{k}$ is the index of $\tau_{k}$, the following statements holds:
  
  \begin{enumerate}
   \item If $d_{k} =\bot$, then $d_{k+1}=\bot$. 
	 \label{item:lemma-index_neg}
   \item If the rule with the conclusion $\Gamma_{p+k}\fCenter\Delta_{p+k}$ is 
	 \rulename{Weak} or \rulename{Subst},
	 then $d_{k+1}=d_{k}$ or $d_{k+1}=\bot$.
	 \label{item:lemma-index_regress}
   \item If the rule with the conclusion $\Gamma_{p+k}\fCenter\Delta_{p+k}$ is 
	 \ruleEqLa or \rulename{$\mathindFromSTosy$ R${}_{\text{2}}$},
	 then $d_{k+1}=d_{k}$.
	 \label{item:lemma-index_not_change}
   \item Assume the rule with the conclusion $\Gamma_{p+k}\fCenter\Delta_{p+k}$ is 
	 \rulename{Case $\mathindToEFromsy$}. 
	 \label{item:lemma-index_Case}
	 \begin{enumerate}
	  \item If $\Gamma_{p+k+1}\fCenter\Delta_{p+k+1}$ is the left assumption of the rule,
		then $d_{k+1}=d_{k}$.
		\label{item:lemma-index_left_ass}
	  \item If $\Gamma_{p+k+1}\fCenter\Delta_{p+k+1}$ is the right assumption of the rule and
		$\tau_{k}$ is not a progress point of the trace,
		then $d_{k+1}=d_{k}$.
		\label{item:lemma-index_not_progress}
	  \item If $\Gamma_{p+k+1}\fCenter\Delta_{p+k+1}$ is the right assumption of the rule and
		$\tau_{k}$ is a progress point of the trace,
		then $d_{k+1}=d_{k}+1$.
		\label{item:lemma-index_progress}
 	 \end{enumerate}
  \end{enumerate}
 \end{lem}

 \begin{proof}[Sketch of proof]%
  Let $\tau_{k}$ be $\mathindToEFrom{t_{k}}$.

  \noindent \cref{item:lemma-index_neg}
  It suffices to show that $t_{k+1} \not\mathdeprel{\Gamma_{p+k+1}} \mathstart$ holds
  if $t_{k} \not\mathdeprel{\Gamma_{p+k}} \mathstart$.
  We can prove it by 
  considering cases according to the rule with the conclusion $\Gamma_{p+k}\fCenter\Delta_{p+k}$.
  For more details, see \cref{appendix:lemma-index}.
  
  \noindent \cref{item:lemma-index_regress} \cref{item:lemma-index_not_change} \cref{item:lemma-index_Case}
  Straightforward.
  For more details, see \cref{appendix:lemma-index}.
 \end{proof}
 
 \begin{lem}%
  \label{lemma:key_lemma}
  For an infinite unfinished path $\mleft(\Gamma_{i}\fCenter\Delta_{i}\mright)_{i\geq 0}$
  in $\mathtreeunfoldingnb{\mathproofcf}$,
  there exists $l\in\mathnat$ such that the following conditions hold:
  \begin{enumerate}
   \item $\Gamma_{l}\fCenter\Delta_{l}$ is a switching point in $\mathtreeunfoldingnb{\mathproofcf}$, and
   \item $\Gamma_{l+1}\fCenter\Delta_{l+1}$ is the right assumption of the rule with the conclusion
	 $\Gamma_{l}\fCenter\Delta_{l}$.
  \end{enumerate}
 \end{lem}
 \begin{proof}
  Since $\mleft(\Gamma_{i}\fCenter\Delta_{i}\mright)_{i\geq 0}$ is an infinite path
  and $\mathtreeunfoldingnb{\mathproofcf}$ satisfies the global trace condition,
  there exists an infinitely progressing trace following a tail of the path.
  Let $\mleft( \tau_{k} \mright)_{k \geq 0}$ be an infinitely progressing trace following 
  $\mleft(\Gamma_{i}\fCenter\Delta_{i}\mright)_{i\geq p}$.
  Let $d_{k}$ be the index of $\tau_{k}$ in $\Gamma_{p+k}$. 

  We show that there exists $l \in \mathnat$ such that $d_{l}=\bot$.
  The set $\mathsetintension{d_{k}}{k\geq 0}$ is finite 
  since the set of sequents in $\mleft(\Gamma_{i}\fCenter\Delta_{i}\mright)_{i\geq 0}$ is finite and
  we have a unique index of an atomic formula with $\mathindToEFromsy$ in $\Gamma_{i}\fCenter\Delta_{i}$. 
  Since $\mleft( \tau_{k} \mright)_{k \geq 0}$ is an infinitely progressing trace 
  following $\mleft(\Gamma_{i}\fCenter\Delta_{i}\mright)_{i\geq p}$, 
  if there does not exist $k' \in \mathnat$ such that $d_{k'}=\bot$,
  \cref{lemma:index} implies that $\mathsetintension{d_{k}}{k\geq 0}$ is infinite.
  Thus, there exists $k' \in \mathnat$ such that $d_{k'}=\bot$. 

  Since $\mleft( \tau_{k} \mright)_{k \geq 0}$ is an infinitely progressing trace following 
  $\mleft(\Gamma_{i}\fCenter\Delta_{i}\mright)_{i\geq p}$, 
  there exists a progress point $\tau_{l}$ with $l>k'$.
  By \cref{lemma:index}, $d_{l}=\bot$.
  Since $\tau_{k}$ is a progress point, 
  $\Gamma_{p+k}\fCenter\Delta_{p+k}$ is a switching point and  
  $\Gamma_{p+k+1}\fCenter\Delta_{p+k+1}$ is the right assumption of the rule.
 \end{proof}

 \begin{definition}[Rightmost path]
  For a derivation tree $\mathprooffig{D}$ and a node $\sigma$ in $\mathprooffig{D}$,
  we define the \emph{rightmost path} from the node $\sigma$ 
  as the path $\mleft( \sigma_{i} \mright)_{0\leq i< \alpha}$ satisfying the following conditions:
  \begin{enumerate}
   \item The node $\sigma_{0}$ is $\sigma$.
   \item If $\sigma_{i}$ is the conclusion of \rulename{Case $\mathindToEFromsy$},
	 the node $\sigma_{i+1}$ is the right assumption of the rule.
   \item If $\sigma_{i}$ is the conclusion of the rules 
	 \rulename{Weak}, \rulename{Subst}, \ruleEqLa, or \rulename{$\mathindFromSTosy$ R${}_{\text{2}}$},
    	 the node $\sigma_{i+1}$ is the assumption of the rule.
  \end{enumerate}
 \end{definition}

 \begin{lem}%
  \label{lemma:rightmost} 
  The rightmost path from a root-like sequent in $\mathtreeunfoldingnb{\mathproofcf}$ is infinite.
 \end{lem}

 \begin{proof}
  By \cref{def:bad-path},
  the rightmost path from a root-like sequent in $\mathtreeunfoldingnb{\mathproofcf}$ is an unfinished path.
  By \cref{lemma:invariant}, every sequent on the path is a root-like sequent.
  By \cref{def:invariant}, \rulename{$\mathindFromSTosy$ R${}_{\text{1}}$} does not occur in the path.  
  Thus, the path is infinite.
 \end{proof}

 \begin{rem}
  For an infinite path in $\mathtreeunfoldingnb{\mathproofcf}$,
  the corresponding path in $\mathdertreecf$ has a bud.
 \end{rem}

 % \subsection{Proof of main theorem} \label{subsec:proof}
 We now have the lemmas to prove \cref{thm:main} \cref{item:thm-main-not-cut-free-provable}.

 \renewcommand{\theenumi}{\textup{(}\roman{enumi}\textup{)}}
 \renewcommand{\labelenumi}{\textup{(}\roman{enumi}\textup{)}}
 \begin{proof}[Proof of \cref{thm:main} \cref{item:thm-main-not-cut-free-provable}]
  We show that there exists a sequence $\mleft(\tilde{\sigma}_{i}\mright)_{i\in\mathnat}$ 
  of switching points in $\mathdertreecf$ which satisfies the following conditions:
  \begin{enumerate}
   \item The height of $\tilde{\sigma}_{i}$ is greater than the height of $\tilde{\sigma}_{i-1}$ in $\mathdertreecf$ for $i>0$.
	 \label{item:ci_height}
   \item For any node $\sigma$ on the path from the root to $\tilde{\sigma}_{i}$ in $\mathdertreecf$ excluding $\tilde{\sigma}_{i}$,
	 $\sigma$ is a switching point if and only if
	 the child of $\sigma$ on the path is
	 the left assumption of the rule \rulename{Case $\mathindToEFromsy$}.%
	 \label{item:ci_is_only_switching}
  \end{enumerate}

  We construct $\mleft(\tilde{\sigma}_{i}\mright)_{i\in\mathnat}$
  and show \cref{item:ci_height} and \cref{item:ci_is_only_switching} by induction on $i$.

  We consider the case $i=0$.

  The rightmost path in $\mathtreeunfoldingnb{\mathproofcf}$ from the root is an infinite unfinished path 
  since $\mathindToEFrom{\mathstart} \fCenter \mathindFromSTo{\mathend}$ is a root-like sequent and 
  there exists no node which is the left assumption of \rulename{Case $\mathindToEFromsy$} on the path.
  By \cref{lemma:key_lemma}, there exists a switching point on the path.
  Hence, there exists a switching point on the rightmost path from the root in $\mathdertreecf$.
  Let $\tilde{\sigma}_{0}$ be the switching point of the smallest height among such switching points.
  \cref{item:ci_height} and \cref{item:ci_is_only_switching} follow immediately for $\tilde{\sigma}_{0}$.

  We consider the case $i>0$.

  Let $\alpha$ be the left assumption of the rule with the conclusion $\tilde{\sigma}_{i-1}$.
  Because of \cref{item:ci_is_only_switching},
  the path from the root to $\tilde{\sigma}_{i-1}$ is also an unfinished path.
  Since $\tilde{\sigma}_{i-1}$ is a switching point, 
  the path from the root to $\alpha$ is also an unfinished path.
  By \cref{lemma:invariant}, $\alpha$ is a root-like sequent.
  By \cref{lemma:rightmost},
  the rightmost path from $\alpha$ in $\mathtreeunfoldingnb{\mathproofcf}$ is infinite.
  Therefore, there is a bud on the rightmost path in $\mathdertreecf$ from $\alpha$. Let $\mu$ be the bud.
  
  Let $\pi_{1}$ be the path from the root to $\mu$ in $\mathdertreecf$ and
  $\pi_{2}$ be the path from $\mathof{\mathcompanioncf}{\mu}$ to $\mu$ in $\mathdertreecf$.
  We define the path $\pi$ in $\mathtreeunfoldingnb{\mathproofcf}$ as $\pi_{1}{\pi_{2}}^{\omega}$.
  Let $\mleft(\sigma_{i}\mright)_{0\leq i}$ be $\pi$.
  Because of \cref{item:ci_is_only_switching}, $\pi$ is an unfinished path.
  By \cref{lemma:key_lemma}, 
  there is a switching point $\sigma_{l}$ and $\sigma_{l+1}$ is the right assumption of the rule.
  Hence, there is a switching point on $\pi_{1}\pi_{2}$ in $\mathdertreecf$ such that
  its child on $\pi_{1}\pi_{2}$ is the right assumption of the rule.
  Define $\tilde{\sigma}_{i}$ as the switching point of the smallest height among such switching points.

  We show $\tilde{\sigma}_{i}$ satisfies the conditions \cref{item:ci_height} and \cref{item:ci_is_only_switching}.

  \cref{item:ci_height}
  By the definition of $\tilde{\sigma}_{i}$, $\tilde{\sigma}_{i}$ is on $\pi_{1}$.
  By the condition \cref{item:ci_is_only_switching}, $\tilde{\sigma}_{i}$ is not on the path from the root to $\tilde{\sigma}_{i-1}$.
  Hence, the height of $\tilde{\sigma}_{i}$ is greater than that of $\tilde{\sigma}_{i-1}$.

  \cref{item:ci_is_only_switching}
  Let $\sigma$ be a node on the path from the root to $\tilde{\sigma}_{i}$ excluding $\tilde{\sigma}_{i}$.
  We can assume $\sigma$ is on the path from $\tilde{\sigma}_{i-1}$ to $\tilde{\sigma}_{i}$ excluding $\tilde{\sigma}_{i}$
  by the induction hypothesis.

  The ``only if'' part:
  Assume that $\sigma$ is a switching point.
  By the definition of $\tilde{\sigma}_{i}$, we see that $\sigma$ is $\tilde{\sigma}_{i-1}$.
  The child of $\tilde{\sigma}_{i-1}$ on the path from the root to $\tilde{\sigma}_{i}$ is $\alpha$,
  which is the left assumption of the rule.

  The ``if'' part:
  Assume that the child of $\sigma$ on the path is the left assumption of the rule.
  Since there is not the left assumption of a rule on the path from $\alpha$ to $\tilde{\sigma}_{i}$,
  we see that $\sigma$ is $\tilde{\sigma}_{i-1}$.
  Thus, $\sigma$ is a switching point.
  
  We complete the construction and the proof of the properties.

  Because of \cref{item:ci_height},  $\tilde{\sigma}_{0}, \tilde{\sigma}_{1}, \ldots$ are all distinct in $\mathdertreecf$.
  Thus, $\mathsetintension{\tilde{\sigma}_{i}}{i\in\mathnat}$ is infinite.
  It is a contradiction since the set of nodes in $\mathdertreecf$ is finite.
 \end{proof}
 
 \subsection{Corollaries of \cref{thm:main}}
 By \cref{thm:main}, we have some corollaries.
 We write \CLKID\ for
 the cyclic proof system for first-order logic with inductive definitions proposed by
 Brotherston and Simpson \cite{Brotherston2011}.
 \begin{cor} 
  \label{cor:main}
  The following statements hold:
  \begin{enumerate}
   \item $\mathindToEFrom{\mathstart} \fCenter \mathindFromSTo{\mathend}$ is provable in \CLKID.
	 \label{item:cor-main-provable}
   \item $\mathindToEFrom{\mathstart} \fCenter \mathindFromSTo{\mathend}$ is not cut-free provable 
	 in \CLKID.
	 \label{item:cor-main-not-cut-free-provable}
  \end{enumerate}
 \end{cor}

 \begin{proof}
  \noindent \cref{item:cor-main-provable}
  The derivation tree in \cref{fig:counter-ex} is also a \CLKID\ proof of the sequent.
  Thus, we have the statement.

  \noindent \cref{item:cor-main-not-cut-free-provable}
  Assume there exists a \CLKID\ cut-free proof of
  $\mathindToEFrom{\mathstart} \fCenter \mathindFromSTo{\mathend}$.
  We write \CLKIDa\ for the cyclic proof system obtained by replacing \ruleEqL\ with \ruleEqLa\ in \CLKID.
  Since \ruleEqL\ is derivable in \CLKIDa, there exists a \CLKIDa\ cut-free proof of
  $\mathindToEFrom{\mathstart} \fCenter \mathindFromSTo{\mathend}$.
  Let $\mleft(\mathdertreecf', \mathcompanioncf'\mright)$ be such a proof.
  Since the rules which can occur in $\mleft(\mathdertreecf', \mathcompanioncf'\mright)$
  are \rulename{Weak}, \rulename{Subst}, \ruleEqLa, 
  \rulename{$\mathindFromSTosy$ R${}_{\text{1}}$}, \rulename{$\mathindFromSTosy$ R${}_{\text{2}}$},
  and \rulename{Case $\mathindToEFromsy$}, we understand
  $\mleft(\mathdertreecf', \mathcompanioncf'\mright)$ as
  a cut-free \CBaseOm\ proof of $\mathindToEFrom{\mathstart} \fCenter \mathindFromSTo{\mathend}$.
  It contradicts \cref{thm:main} \cref{item:thm-main-not-cut-free-provable}.
 \end{proof}

 \cref{cor:main} means that $\mathindToEFrom{\mathstart} \fCenter \mathindFromSTo{\mathend}$ is
 a counterexample to cut-elimination in \CLKID, and therefore we have the following corollary.
 \begin{cor} 
  \label{cor:cut-elimination-and-arity}
  We do not eliminate the cut rule in \CLKID\
  if we restrict predicates in the language to unary predicates.
 \end{cor}

%#!latexmk -c -gg -lualatex main.tex
 \section{Conclusions and discussion}
 \label{sec:conclusion}
 We have shown that $\mathindToEFrom{\mathstart} \fCenter \mathindFromSTo{\mathend}$
 is a counterexample with only unary inductive predicates to cut-elimination in \CBaseOm.
 This counterexample implies that
 we cannot eliminate the cut rule in first-order logic with inductive definitions
 if we restrict predicates in the language to unary predicates and $=$.

 \begin{figure}[t]
  \begin{prooftree}
   \small
   \AxiomC{}
   \DeduceC{} %(1)
   
   \AxiomC{}
   \DeduceC{} %(2)
    
   \AxiomC{($\heartsuit$) $\underline{\mathindToEFrom{\mathnext x}}, \mathindToEFrom{\mathstart} \fCenter \mathindFromSTo{\mathend}$}
   \UnaryInfC{$\mathnext x = y, \underline{\mathindToEFrom{\mathnext y}}, \mathindToEFrom{\mathstart} \fCenter \mathindFromSTo{\mathend}$} %(2)
  
   \rulenamelabel{Case $\mathindToEFromsy$} %(2)
   \BinaryInfC{($\heartsuit$) $\underline{\mathindToEFrom{\mathnext x}}, \mathindToEFrom{\mathstart} \fCenter \mathindFromSTo{\mathend}$}
   \UnaryInfC{$\mathstart = x, \mathindToEFrom{\mathnext x}, \mathindToEFrom{\mathstart} \fCenter \mathindFromSTo{\mathend}$} %(1)
  
   \rulenamelabel{Case $\mathindToEFromsy$} %(1)
   \BinaryInfC{$\mathindToEFrom{\mathstart} \fCenter \mathindFromSTo{\mathend}$}
  \end{prooftree}
  \caption{There can be an infinitely progressing trace in the rightmost path} 
  \label{fig:compare}
 \end{figure}

 The proofs for counterexamples to cut-elimination in \cite{Masuoka2021} and this paper is
 more complicated than in \cite{Kimura2020} since there is the left-contraction rule.
 The proof technique in \cite{Kimura2020} is to show that
 there does not exist a companion in the right-most path if there exists an infinitely progressing trace.
 However, since there is the left-contraction rule, 
 we have an infinitely progressing trace in the right-most path, as in \cref{fig:compare}.
 Comparing our proof technique with the proof technique in \cite{Saotome2021} is reserved for future work.

 \begin{table}[tb]
  \centering
  \begin{tabular}{|c|p{8em}|p{8em}|p{8em}|p{7em}|}
   \hline
   &
       nullary \newline predicates
       &
	   unary \newline  predicates
	   &
	       binary \newline predicates
	       &
		   $N$-ary ($N\geq 3$) \newline  predicates
		   \\
    \hline
    Classical Logic 
    &
    ? (perhaps Yes)
    &
    No$^*$ (This paper)
    &
    No$^*$ (This paper)
    &
    No$^*$~\cite{Masuoka2021}
    \\
    Separation Logic
    &
    No (by \cite{Saotome2021})
    &
    No (by \cite{Saotome2021})
    &
    No~\cite{Kimura2020}
    &
    No~\cite{Kimura2020}
    \\    
    Bunched Logic
    &
    No~\cite{Saotome2021}
    &
    No~\cite{Saotome2021}
    &
    No~\cite{Saotome2021}
    &
    No~\cite{Saotome2021}
    \\
    \hline    
  \end{tabular}
   \caption{The arity of inductive predicates and the cut-elimination property in each cyclic proof system for some logics} 
   \label{table:discussion}
 \end{table}

 \cref{table:discussion} shows
 the results we obtained about the cut-elimination property of each cyclic proof system for some logics. 
 ``No'' means that the cut-elimination property does not hold.
 ``No$^*$'' means that the cut-elimination property does not hold
 if there are constants and a unary function symbol.
 The second and third column results in the ``Separation Logic'' raw are easily obtained from
 Saotome's result \cite{Saotome2021}
 because the counterexample for the cyclic proof system of bunched logic also works for separation logic. 

 Why does not the cut-elimination property hold in cyclic proof systems? 
 The reason is not yet wholly understood, but we discuss it briefly.
 The proofs for counterexamples to cut-elimination in \cite{Kimura2020, Saotome2021, Masuoka2021}
 and this paper have one thing in common.
 It is to contradict the finiteness of the sequent occurring in the cut-free proof of each counterexample
 if it exists.
 The more important fact is that the cut-elimination property of
 \LKIDOm, which is obtained by expanding the shape of each proof figure into an infinitary tree, holds.
 These facts suggest that the reason the cut-elimination property does not hold in cyclic proof systems
 is the finiteness of occurring sequents in each proof.

 Now, we discuss the ``Classical Logic'' raw in \cref{table:discussion}. 
 It suggests the reason the cut-elimination property does not hold 
  might be a unary function symbol in the language.
  Then, we conjecture that the cut-elimination property of \CLKID\ hold
  if there is no unary function symbol, and therefore 
  the cut-elimination property of \CLKID\ hold if restricting the arity of predicates to nullary.

 By the way, can we restrict cut formulas in \CLKID\ without changing provability?
 Saotome et al. \cite{Saotome2020} suggest that we cannot restrict the cut formulas
 to formulas presumable from the goal sequent 
 in the cyclic proof system for symbolic-heaps, a fragment of separation logic. 
 The cut formulas in \cref{fig:counter-ex} are presumable.
 Also, the cut formula in a \CLKID\ proof of the counterexample in \cite{Masuoka2021} is presumable.
 Can we restrict the cut formulas to presumable formulas?
 If the answer to the problem is ``Yes'', there may be an efficient proof search in \CLKID. 
 Research into solving the problem is in progress.
%#!latexmk -c -gg -lualatex main.tex
 \section*{Acknowledgments}
 % We thank the anonymous referees for the reviews.
 We are grateful to Koji Nakazawa and Kenji Saotome 
 for giving us helpful suggestions from the early stage of this work.
 We would also like to thank Makoto Tatsuta for giving us valuable comments.
 This work was supported by the Research Institute for Mathematical Sciences,
 an International Joint Usage/Research Center located in Kyoto University.

\newpage
\appendix

\section{The proof of the cycle normalization of $\CBaseOm$ (\cref{prop:cyclenormalization})}

\begin{proof}[Proof of \cref{prop:cyclenormalization}]
  Let $(\DerivTree,\BudCompanion)$ be $\ProofTree$, and $D_1$ be the tree unfolding of $\ProofTree$.
  We write $\sigma \segeq \sigma'$ when $\sigma$ is an initial segment of $\sigma'$.
  If $\sigma$ is a strict initial segment of $\sigma'$, we write $\sigma \seg \sigma'$. 
  We define $\DerivTree^{(\sigma)}$ by $\DerivTree^{(\sigma)}(\sigma_1) = \DerivTree(\sigma \sigma_1)$,
  $\Bar{S}$ as $\{ \sigma' \ |\ \sigma' \segeq \sigma \in S \}$,
  and
  $S^\circ$ as $\{ \sigma' \ |\ \sigma' \seg \sigma \in S \}$.

  Then define $S_1$ and $S_2$ by: 
  \begin{align*}
    S_1 &=
    \{ \sigma \in \Dom{\DerivTree_1} \ |\ \exists \sigma' \seg \sigma
    (\DerivTree_1^{(\sigma)} = \DerivTree_1^{(\sigma')}),
    \forall \sigma_1, \sigma_2 \seg \sigma(\DerivTree_1^{(\sigma_1)} \neq \DerivTree_1^{(\sigma_2)}),
    \forall n\exists \sigma_1 \geseq \sigma(\sigma_1 \in \Dom{\DerivTree_1}, |\sigma_1|\ge n) \},
    \\
    S_2 &= \{ \sigma \in \Dom{\DerivTree_1} \ |\ \sigma 0 \notin \Dom{\DerivTree_1}, 
    \forall \sigma' \segeq \sigma(\sigma' \notin S_1) \}.
  \end{align*}

  $S_1$ is the set of nodes such that
  the node is on some infinite path and
  the node is of the smallest height on the path
  among nodes, each of which has some inner node 
  of the same subtree.
  $S_2$ is the set of leaf nodes of finite paths
  which are not cut by $S_1$.

  Define $\DerivTree'$ by
  \[
  \DerivTree'(\sigma) =
  \mleft\{
  \begin{array}{ll}
    \DerivTree_1(\sigma) 
    &
    \hbox{if $\sigma \in (S_1)^\circ \cup \Bar{S_2}$},
    \\
    (\Gamma \vdash \Delta,\ruleBud)
    &
    \hbox{if $\sigma \in S_1$ and $\DerivTree_1(\sigma) = (\Gamma \vdash \Delta, R)$}.
  \end{array}
  \mright.
  \]

  Define $\BudCompanion'$ by $\BudCompanion'(\sigma) = \sigma'$ for $\sigma \in \Bud{\DerivTree}$
  where $\sigma' \seg \sigma$ and $\DerivTree_1^{(\sigma)} = \DerivTree_1^{(\sigma')}$.

  We can show that $\Dom{\DerivTree'}$ is finite as follows.
  Since $\Dom{\DerivTree'} = \Bar{S_1} \cup \Bar{S_2}$,
  we have $\Dom{\DerivTree'} \subseteq \Dom{\DerivTree_1}$.
  Since $\DerivTree_1$ is finite-branching, $\DerivTree'$ is so.
  Assume $\Dom{\DerivTree'}$ is infinite to show contradiction.
  By K\"onig's lemma, there is some infinite path $(\sigma_i)_i$
  such that $\sigma_i \in \Dom{\DerivTree'}$.
  Since $\DerivTree_1$ is regular, the set $\{ \DerivTree_1^{(\sigma_i)} \}_i$ is finite.
  Hence there are $j<k$ such that $\DerivTree_1^{(\sigma_j)} = \DerivTree_1^{(\sigma_k)}$.
  Take the smallest $k$ among such $k$'s.
  Then $\sigma_k \in S_1$.
  Hence $\sigma_{k+1} \notin \Bar{S_1}$.
  Hence $\sigma_{k+1} \notin \Dom{\DerivTree'}$, which contradicts.

Then $(\DerivTree',\BudCompanion')$ is a \CLKID\ cycle-normal pre-proof.

Define $\DerivTree_1'$ as the tree-unfolding of $(\DerivTree',\BudCompanion')$.

We can show $\DerivTree_1=\DerivTree_1'$ on $\Dom{\DerivTree_1'}$ as follows.

Case 1 where for any $\sigma' \segeq \sigma$, $\sigma' \notin S_1$.
$\DerivTree_1'(\sigma) = \DerivTree'(\sigma) = \DerivTree_1(\sigma)$.

Case 2 where 
there is some $\sigma_1 \segeq \sigma$ such that $\sigma_1 \in S_1$.
Let $\sigma_1 \sigma_2$ be $\sigma$ and $\sigma_3$ be $\BudCompanion'(\sigma_1)$.
Then $\DerivTree_1(\sigma) = \DerivTree_1^{(\sigma_1)}(\sigma_2) = \DerivTree_1^{(\sigma_3)}(\sigma_2)
= \DerivTree_1(\sigma_3 \sigma_2) = \DerivTree_1'(\sigma_3 \sigma_2)$ by the induction hypothesis,
it is $\DerivTree_1'(\sigma_1 \sigma_2)$ by definition of $\DerivTree_1'$, 
and it is $\DerivTree_1'(\sigma)$.

We show $\Dom{\DerivTree_1} \subseteq \Dom{\DerivTree_1'}$ as follows.
By induction on $|\sigma|$, we will show
$\sigma \in \Dom{\DerivTree_1}$ implies $\sigma \in \Dom{\DerivTree_1'}$.
If $\sigma \in S_1^\circ \cup \Bar{S_2}$, then
$\sigma \in \Dom{\DerivTree'} - S_1$.
Hence $\sigma \in \Dom{\DerivTree_1'}$.
If there is some $\sigma_1 \seg \sigma$ such that $\sigma_1 \in S_1$,
then by letting $\sigma = \sigma_1 \sigma_2$ and $\sigma_3 = \BudCompanion'(\sigma_1)$,
$\DerivTree_1(\sigma) = \DerivTree_1(\sigma_3 \sigma_2)$ by definition of $\BudCompanion'$,
by the induction hypothesis for $\sigma_3 \sigma_2$ it is $\DerivTree_1'(\sigma_3 \sigma_2)$,
and it is $\DerivTree_1'(\sigma)$ by definition of $\DerivTree_1'$.
Thus we have shown $\Dom{\DerivTree_1} \subseteq \Dom{\DerivTree_1'}$.
Hence $\DerivTree_1=\DerivTree_1'$ holds.
\end{proof}

\section{The proofs of \cref{lemma:invariant} and \cref{lemma:index}}
In this appendix, we show \cref{lemma:invariant} and \cref{lemma:index}.
We assume terms in this appendix are of the form
$\mathnext^{n}\mathstart$, $\mathnext^{n}\mathend$, or $\mathnext^{n}x$ with some variable $x$.
For terms $t_{1}$ and $t_{2}$,  we write $t_{1}\equiv t_{2}$ if $t_{1}$ is the same as $t_{2}$.
For sets of formulas $\Gamma_{1}$ and $\Gamma_{2}$,  
we write $\Gamma_{1}\equiv \Gamma_{2}$ if $\Gamma_{1}$ is the same as $\Gamma_{2}$.

  \subsection{The lemmas for \cref{lemma:invariant} and \cref{lemma:index}}
  We show the lemmas for \cref{lemma:invariant} and \cref{lemma:index}.

  \begin{lem}%
  \label{lemma:subst_rel}
  Let $\Gamma$ be a set of formulas and $\theta$ be a substitution.
  \begin{enumerate}
   \item For any terms $t_{1}$ and $t_{2}$,
	 $t_{1}\mleft[\theta\mright] \mathequivrel{\Gamma\mleft[\theta\mright]} t_{2}\mleft[\theta\mright]$ 
	 if $t_{1} \mathequivrel{\Gamma} t_{2}$. 
	 \label{item:lemma-subst_rel-equiv}
   \item For any terms $t_{1}$ and $t_{2}$,
	 $t_{1} \not\mathdeprel{\Gamma} t_{2}$
	 if $t_{1}\mleft[\theta\mright] \not\mathdeprel{\Gamma\mleft[\theta\mright]} t_{2}\mleft[\theta\mright]$. 
	 \label{item:lemma-subst_rel-deprel}
  \end{enumerate}
  \end{lem}
 
 \begin{proof}
  \cref{item:lemma-subst_rel-equiv}
  We prove the statement by induction on the definition of $\mathequivrel{\Gamma}$. 
  We only show the base case.
  Assume ${t_{1}=t_{2}}\in\Gamma$.
  Then, ${t_{1}\mleft[\theta\mright]=t_{2}\mleft[\theta\mright]}\in\Gamma\mleft[\theta\mright]$.
  Thus, $t_{1}\mleft[\theta\mright] \mathequivrel{\Gamma\mleft[\theta\mright]} t_{2}\mleft[\theta\mright]$.

  \cref{item:lemma-subst_rel-deprel}
  By  \cref{def:deprel} and \cref{item:lemma-subst_rel-equiv}, we have the statement.
 \end{proof}

 \begin{lem}%
  \label{lemma:eq_rel}
  Let $\Gamma$ be a set of formulas, $u_{1}$, $u_{2}$ be terms, $v_{1}$, $v_{2}$ be variables, 
  $\Gamma_{1} \equiv \mleft(\Gamma\mleft[v_{1}:=u_{1}, v_{2}:=u_{2}\mright], u_{1}=u_{2}\mright)$, and 
  $\Gamma_{2} \equiv \mleft(\Gamma\mleft[v_{1}:=u_{2}, v_{2}:=u_{1}\mright], u_{1}=u_{2}\mright)$.
  \begin{enumerate}
   \item For any terms $t_{1}$ and $t_{2}$,
	 $t_{1}\mleft[v_{1}:=u_{1}, v_{2}:=u_{2}\mright] \mathequivrel{\Gamma_{1}} t_{2}\mleft[v_{1}:=u_{1}, v_{2}:=u_{2}\mright]$ 
	 if  
	 $t_{1}\mleft[v_{1}:=u_{2}, v_{2}:=u_{1}\mright] \mathequivrel{\Gamma_{2}} t_{2}\mleft[v_{1}:=u_{2}, v_{2}:=u_{1}\mright]$. 
	 \label{item:lemma-eq_rel-equiv}
   \item For any terms $t_{1}$ and $t_{2}$,
	 $t_{1}\mleft[v_{1}:=u_{2}, v_{2}:=u_{1}\mright] \not\mathdeprel{\Gamma_{2}} t_{2}\mleft[v_{1}:=u_{2}, v_{2}:=u_{1}\mright]$ 
	 if 
	 $t_{1}\mleft[v_{1}:=u_{1}, v_{2}:=u_{2}\mright] \not\mathdeprel{\Gamma_{1}} t_{2}\mleft[v_{1}:=u_{1}, v_{2}:=u_{2}\mright]$. \label{item:lemma-eq_rel-deprel}
  \end{enumerate}
 \end{lem}

 \begin{proof}%
  \cref{item:lemma-eq_rel-equiv}
  We prove the statement by induction on the definition of $\mathequivrel{\Gamma_{2}}$.
  We only show the base case.
  Assume
  ${t_{1}\mleft[v_{1}:=u_{2}, v_{2}:=u_{1}\mright]=t_{2}\mleft[v_{1}:=u_{2}, v_{2}:=u_{1}\mright]} \in \Gamma_{2}$
  to show $t_{1}\mleft[v_{1}:=u_{1}, v_{2}:=u_{2}\mright] \mathequivrel{\Gamma_{1}} t_{2}\mleft[v_{1}:=u_{1}, v_{2}:=u_{2}\mright]$.
  If $t_{1}\mleft[v_{1}:=u_{2}, v_{2}:=u_{1}\mright]=t_{2}\mleft[v_{1}:=u_{2}, v_{2}:=u_{1}\mright]$
  is $u_{1}=u_{2}$, then
  $t_{1}=t_{2}$ is $v_{2}=v_{1}$, $v_{2}=u_{2}$, $u_{1}=v_{1}$, or $u_{1}=u_{2}$.
  Therefore,
  $t_{1}\mleft[v_{1}:=u_{1}, v_{2}:=u_{2}\mright] \mathequivrel{\Gamma_{1}} t_{2}\mleft[v_{1}:=u_{1}, v_{2}:=u_{2}\mright]$.

  Assume $t_{1}\mleft[v_{1}:=u_{2}, v_{2}:=u_{1}\mright]=t_{2}\mleft[v_{1}:=u_{2}, v_{2}:=u_{1}\mright]$
  is not $u_{1}=u_{2}$.
  By case analysis, we have $t_{1}=t_{2}\in\Gamma$. Hence,
  ${t_{1}\mleft[v_{1}:=u_{1}, v_{2}:=u_{2}\mright]=t_{2}\mleft[v_{1}:=u_{1}, v_{2}:=u_{2}\mright]}\in \Gamma_{1}$. 
  Therefore, we have
  $t_{1}\mleft[v_{1}:=u_{1}, v_{2}:=u_{2}\mright]\mathequivrel{\Gamma_{1}}t_{2}\mleft[v_{1}:=u_{1}, v_{2}:=u_{2}\mright]$.

  \cref{item:lemma-eq_rel-deprel}
  By \cref{def:deprel} and \cref{item:lemma-eq_rel-equiv}, we have the statement.
 \end{proof}

 \begin{lem}%
  \label{lemma:strcl}
  For a set of formulas $\Gamma$,
  the following statements are equivalent:
  \begin{enumerate}
   \item $u_{1}\mathequivrel{\Gamma}u_{2}$. \label{item:lemma-strcl-equivrel}
   \item There exists a finite sequence of terms 
	 $\mleft( t_{i} \mright)_{0\leq i \leq n}$ with $n\geq 0$ such that 
	 $t_{0}\equiv u_{1}$, $t_{n}\equiv u_{2}$ and ${t_{i}=t_{i+1}}\in\mleft[\Gamma\mright]$ 
	 for $0\leq i < n$, where
	 \[
	 \mleft[\Gamma\mright]=\mathsetintension{\mathnext^{n}t_{1}=\mathnext^{n}t_{2}}{
	 n\in\mathnat \text{ and either }
	 t_{1}=t_{2}\in \Gamma \text{ or } t_{2}=t_{1}\in\Gamma}.
	 \]
	 \label{item:lemma-strcl}
  \end{enumerate}
 \end{lem}

 \begin{proof}%
  \cref{item:lemma-strcl-equivrel} $\Rightarrow$ \cref{item:lemma-strcl}:
  Assume $u_{1}\mathequivrel{\Gamma}u_{2}$
  to prove \cref{item:lemma-strcl} by induction on the definition of  $\mathequivrel{\Gamma}$.
  We consider cases according to the clauses of the definition. 

  Case 1.
  If $u_{1}=u_{2}\in \Gamma$, 
  then we have $u_{1}=u_{2}\in \mleft[\Gamma\mright]$.
  Thus, we have \cref{item:lemma-strcl}.

  Case 2.
  If $u_{1}\equiv u_{2}$,
  then we have \cref{item:lemma-strcl}.

  Case 3.
  We consider the case where
  $u_{2}\mathequivrel{\Gamma}u_{1}$.
  By the induction hypothesis,
  there exists a finite sequence of terms $\mleft( t_{i} \mright)_{0\leq i \leq n}$ such that 
  $t_{0}\equiv u_{2}$, $t_{n}\equiv u_{1}$ and ${t_{i}=t_{i+1}}\in\mleft[\Gamma\mright]$ with $0\leq i < n$.
  Let $t'_{i}\equiv t_{n-i}$.
  The finite sequence of terms $\mleft( t'_{i} \mright)_{0\leq i \leq n}$ 
  satisfies $t'_{0}\equiv u_{1}$, $t'_{n}\equiv u_{2}$ and ${t'_{i}=t'_{i+1}}\in\mleft[\Gamma\mright]$.
  Thus, we have \cref{item:lemma-strcl}.

  Case 4.
  We consider the case where $u_{1}\mathequivrel{\Gamma}u_{3}$, $u_{3}\mathequivrel{\Gamma}u_{2}$.
  By the induction hypothesis,
  there exist two finite sequences of terms 
  $\mleft( t_{i} \mright)_{0\leq i \leq n}$, $\mleft( t'_{j} \mright)_{0\leq j \leq m}$ 
  such that $t_{0}\equiv u_{1}$, $t_{n}\equiv t'_{0}\equiv u_{3}$, $t'_{m}\equiv u_{2}$,
  ${t_{i}=t_{i+1}}\in\mleft[\Gamma\mright]$ and ${t'_{j}=t'_{j+1}}\in\mleft[\Gamma\mright]$
  with $0\leq i < n$, $0\leq j < m$.
  Define $\hat{t}_{k}$ as $t_{k}$ if $0\leq k < n$ and $t'_{k-n}$ if $n\leq k \leq n+m$.
  The finite sequence of terms 
  $\mleft( \hat{t}_{k} \mright)_{0\leq k \leq n}$ 
  satisfies
  $\hat{t}_{0}\equiv u_{1}$, $\hat{t}_{n}\equiv u_{2}$ and
  ${\hat{t}_{k}=\hat{t}_{k+1}}\in\mleft[\Gamma\mright]$.
  Thus, we have \cref{item:lemma-strcl}.

  Case 5. 
  We consider the case where 
  $\hat{u}_{1} \mathequivrel{\Gamma} \hat{u}_{2}$,
  $u_{1}\equiv u\mleft[v:=\hat{u}_{1}\mright]$ and $u_{2}\equiv u\mleft[v:=\hat{u}_{2}\mright]$.
  By the induction hypothesis,
  there exists a finite sequence of terms $\mleft( t_{i} \mright)_{0\leq i \leq n}$
  with $n\in\mathnat$ such that $t_{0}\equiv \hat{u}_{1}$, $t_{n}\equiv \hat{u}_{2}$,
  ${t_{i}=t_{i+1}}\in\mleft[\Gamma\mright]$ with $0\leq i < n$.

  Assume $v$ does not occur in $u$.
  In this case,  we have
  $u_{1}\equiv u\mleft[v:=\hat{u}_{1}\mright] \equiv u \equiv u\mleft[v:=\hat{u}_{2}\mright] \equiv u_{2}$.
  Hence, \cref{item:lemma-strcl} holds.

  Assume $v$ occurs in $u$.
  In this case, we have $u\equiv \mathnext^{m}v$ for some natural numbers $m$.
  Let $t'_{i}=\mathnext^{m}t_{i}$ for $0\leq i\leq n$.
  The finite sequence of terms  $\mleft( t'_{i} \mright)_{0\leq i \leq n}$ 
  satisfies $t'_{0}\equiv u_{1}$, $t'_{n}\equiv u_{2}$ and ${t'_{i}=t'_{i+1}}\in\mleft[\Gamma\mright]$.

  \cref{item:lemma-strcl} $\Rightarrow$ \cref{item:lemma-strcl-equivrel}:
  Assume \cref{item:lemma-strcl} to show \cref{item:lemma-strcl-equivrel}.
  By the assumption, there exists a finite sequence of terms $\mleft( t_{i} \mright)_{0\leq i \leq n}$ 
  with $n\in\mathnat$ such that
  $t_{0}\equiv u_{1}$, $t_{n}\equiv u_{2}$ and ${t_{i}=t_{i+1}}\in\mleft[\Gamma\mright]$ with  $0\leq i < n$.
  If ${t_{i}=t_{i+1}}\in\mleft[\Gamma\mright]$, then $t_{i}=t_{i+1}$ is $\mathnext^{n}\hat{t}_{1}=\mathnext^{n}\hat{t}_{2}$, 
  where $\hat{t}_{1}=\hat{t}_{2}\in \Gamma$ or $\hat{t}_{2}=\hat{t}_{1}\in \Gamma$.
  Therefore, $t_{i}\mathequivrel{\Gamma}t_{i+1}$.
  Because of the transitivity of $\mathequivrel{\Gamma}$, we have $u_{1}\mathequivrel{\Gamma}u_{2}$.
 \end{proof}

 For a term $t$, we define $\mathvars{t}$ as a variable or a constant in $t$.

 \begin{lem}%
  \label{lemma:right_asp}
  For a set of formulas $\Gamma_{1}$ and
  $\Gamma_{2} \equiv \mleft(\Gamma_{1}, u=u'\mright)$,
  if $\mathvars{u'}$ do not occur in $\Gamma_{1}, u, t, t'$, 
  then $t \mathequivrel{\Gamma_{2}} t'$  implies $t \mathequivrel{\Gamma_{1}} t'$.
 \end{lem}

 \begin{proof}
  Assume $t \mathequivrel{\Gamma_{2}} t'$ and $t \not\mathdeprel{\Gamma_{1}} \mathvars{u'}$.
  By \cref{lemma:strcl}, 
  there exists a sequence $\mleft( t_{j} \mright)_{0\leq j \leq n}$ with $n\in\mathnat$
  such that $t_{0}\equiv t$, $t_{n}\equiv t'$ and
  ${t_{j}=t_{j+1}}\in\mleft[\Gamma_{2}\mright]$  with $0\leq j < n$.
  We show $t \mathequivrel{\Gamma_{1}} t'$  by induction on $n$.

  For $n=0$, we have $t \mathequivrel{\Gamma_{1}} t'$ immediately.

  We consider the case where $n>0$.

  If $t_{j}\not\equiv \mathnext^{m}u'$ for $0\leq j \leq n$ and $m\in\mathnat$, 
  then ${t_{j}=t_{j+1}}\in\mleft[\Gamma_{1}\mright]$ with $0\leq i < n$.
  By \cref{lemma:strcl},  we have $t \mathequivrel{\Gamma_{1}} t'$.
  
  Assume that there exists $j_{0}$ with $0\leq j_{0} \leq n$,
  such that $t_{j_{0}}\equiv \mathnext^{m}u'$ for $m\in\mathnat$.
  Since any formula of $\mleft[\Gamma_{2}\mright]$ in which $u'$ occurs
  is either $\mathnext^{l}u=\mathnext^{l}u'$ or $\mathnext^{l}u'=\mathnext^{l}u$ with $l\in\mathnat$
  and $\mathvars{u'}$ do not occur in $t, t'$,
  we have $t_{j_{0}-1}\equiv t_{j_{0}+1} \equiv \mathnext^{m}u$.
  Define $\bar{t}_{k}$ as $t_{k}$ if $0\leq k < j_{0}$ and $t_{k+1}$ if $j_{0}\leq k \leq n-1$.
  Then, $\bar{t}_{0}\equiv t$, $\bar{t}_{n-1}\equiv t'$ and
  ${\bar{t}_{k}=\bar{t}_{k+1}}\in\mleft[\Gamma_{2}\mright]$ with $0\leq k < n-1$.
  By the induction hypothesis, we have $t \mathequivrel{\Gamma_{1}} t'$. 
 \end{proof}

 \begin{lem}%
  \label{lemma:left_asp}
  For a set of formulas $\Gamma_{1}$ and 
  $\Gamma_{2} \equiv \mleft(\Gamma_{1}, u = u'\mright)$,
  if $t \not\mathdeprel{\Gamma_{1}} u$ and $t \not\mathdeprel{\Gamma_{1}} u'$,
  then $t \mathequivrel{\Gamma_{2}} t'$ implies $t \mathequivrel{\Gamma_{1}} t'$.
 \end{lem}

 \begin{proof}
  Assume $t \not\mathdeprel{\Gamma_{1}} u$, $t \not\mathdeprel{\Gamma_{1}} u'$, and $t \mathequivrel{\Gamma_{2}} t'$.
  By \cref{lemma:strcl}, 
  there exists a sequence $\mleft( t_{i} \mright)_{0\leq i \leq m}$ with $m\in\mathnat$ such that
  $t_{0}\equiv t$, $t_{m}\equiv t'$ and ${t_{i}=t_{i+1}}\in\mleft[\Gamma_{2}\mright]$ with $0\leq i < m$.
  
  If $t_{i}\not\equiv \mathnext^{l}u$ and $t_{i}\not\equiv \mathnext^{l}u'$ 
  for all $0\leq i \leq n$, and any $l\in\mathnat$,
  then ${t_{i}=t_{i+1}}\in\mleft[\Gamma_{1}\mright]$ with all $0\leq i < m$.
  By \cref{lemma:strcl}, we have $t \mathequivrel{\Gamma_{1}} t'$.
  
  Assume that there exists $i$ with $0\leq i \leq n$,
  such that $t_{i}\equiv \mathnext^{l}u$ or $t_{i}\equiv \mathnext^{l}u'$ for some $l\in\mathnat$.
  Let $i_{0}$ be the least number among such $i$'s.
  Since $i_{0}$ is the least,
  we have ${t_{i}=t_{i+1}}\in\mleft[\Gamma_{1}\mright]$ for all $0\leq i < i_{0}$.
  By \cref{lemma:strcl},
  we have $t\mathequivrel{\Gamma_{1}} \mathnext^{l}u$ or $t\mathequivrel{\Gamma_{1}} \mathnext^{l}u'$.
  This contradicts $t \not\mathdeprel{\Gamma_{1}} u$ and $t \not\mathdeprel{\Gamma_{1}} u'$.
 \end{proof}

  \subsection{The proof of \cref{lemma:invariant}}
  \label{appendix:lemma-invariant}
  We show \cref{lemma:invariant}.

  Let $\mleft( \Gamma_{i} \fCenter \Delta_{i} \mright)_{0\leq i <\alpha}$ be an unfinished path. 
  We prove the statement by the induction on $i$.
  
  For $i=0$, $\Gamma_{0} \fCenter \Delta_{0}$ is a root-like sequent by \cref{def:bad-path}.

  For $i>0$, we consider cases according to the rule with the conclusion $\Gamma_{i-1}\fCenter\Delta_{i-1}$.

  %%%%%%%%%%%%%%%%%%%%%%%%%%%%%%%%%%%%%%%%%

  Case 1. The case \rulename{Weak}.

  \cref{item:def-invariant_start_and_end}
  By the induction hypothesis  \cref{item:def-invariant_start_and_end},
  we have $\mathstart\not\mathdeprel{\Gamma_{i-1}}\mathend$.
  By $\Gamma_{i}\subseteq\Gamma_{i-1}$, we have $\mathstart\not\mathdeprel{\Gamma_{i}}\mathend$.

  \cref{item:def-invariant_existence_of_premise}
  Let $\mathindFromSTo{t}\in\Delta_{i}$.
  By $\Delta_{i}\subseteq\Delta_{i-1}$, we have $\mathindFromSTo{t}\in\Delta_{i-1}$.
  By the induction hypothesis  \cref{item:def-invariant_existence_of_premise},
  $t\not\mathdeprel{\Gamma_{i-1}}\mathstart$.
  By $\Gamma_{i}\subseteq\Gamma_{i-1}$, we have $t\not\mathdeprel{\Gamma_{i}}\mathstart$.

  \cref{item:def-invariant_defined_index}
  Assume $\mathnext^{n}\mathstart\mathequivrel{\Gamma_{i}}\mathnext^{m}\mathstart$.
  By $\Gamma_{i}\subseteq\Gamma_{i-1}$,
  we have $\mathnext^{n}\mathstart\mathequivrel{\Gamma_{i-1}}\mathnext^{m}\mathstart$.
  By the induction hypothesis \cref{item:def-invariant_defined_index}, $n=m$.

  %%%%%%%%%%%%%%%%%%%%%%%%%%%%%%%%%%%%%%%%%%%%%

  Case 2. The case \rulename{Subst} with a substitution $\theta$.
  
  \cref{item:def-invariant_start_and_end}
  By the induction hypothesis  \cref{item:def-invariant_start_and_end},
  we have $\mathstart\not\mathdeprel{\Gamma_{i-1}}\mathend$.
  By \cref{lemma:subst_rel} \cref{item:lemma-subst_rel-deprel},
  we have $\mathstart\not\mathdeprel{\Gamma_{i}}\mathend$.

  \cref{item:def-invariant_existence_of_premise}
  Let $\mathindFromSTo{t}\in\Delta_{i}$.
  By $\Delta_{i-1}\equiv\Delta_{i}\mathsubstbox{\theta}$,
  $\mathindFromSTo{t\mathsubstbox{\theta}}\in\Delta_{i-1}$.
  By the induction hypothesis \cref{item:def-invariant_existence_of_premise},
  $t\mathsubstbox{\theta}\not\mathdeprel{\Gamma_{i-1}}\mathstart$.
  By \cref{lemma:subst_rel} \cref{item:lemma-subst_rel-deprel},
  $t\not\mathdeprel{\Gamma_{i}}\mathstart$.

  \cref{item:def-invariant_defined_index}
  Assume $\mathnext^{n}\mathstart\mathequivrel{\Gamma_{i}}\mathnext^{m}\mathstart$.
  By \cref{lemma:subst_rel} \cref{item:lemma-subst_rel-equiv},
  we have $\mathnext^{n}\mathstart\mathequivrel{\Gamma_{i-1}}\mathnext^{m}\mathstart$ 
  By the induction hypothesis \cref{item:def-invariant_defined_index}, $n=m$.
  
  %%%%%%%%%%%%%%%%%%%%%%%%%%%%%%%%%%%%%%%%%%

  Case 3. The case \ruleEqLa.
  
  Let $u_{1} = u_{2}$ be the principal formula of the rule. There exists $\Gamma$ and $\Delta$ such that
  \begin{align*}
   \Gamma_{i-1} 
   &\equiv
   \mleft(\Gamma\mathsubstbox{\mathsubst{v_{1}}{u_{1}}, \mathsubst{v_{2}}{u_{2}}}, {u_{1} = u_{2}}\mright), \\
   \Delta_{i-1} 
   &\equiv
   \mleft(\Delta\mathsubstbox{\mathsubst{v_{1}}{u_{1}}, \mathsubst{v_{2}}{u_{2}}}, {u_{1} = u_{2}}\mright), \\
   \Gamma_{i} 
   &\equiv
   \mleft(\Gamma\mathsubstbox{\mathsubst{v_{1}}{u_{2}}, \mathsubst{v_{2}}{u_{1}}}, {u_{1} = u_{2}}\mright),
   \text{ and }\\
   \Delta_{i} 
   &\equiv
   \mleft(\Delta\mathsubstbox{\mathsubst{v_{1}}{u_{2}}, \mathsubst{v_{2}}{u_{1}}}, {u_{1} = u_{2}}\mright).
  \end{align*}

  \cref{item:def-invariant_start_and_end}
  By the induction hypothesis  \cref{item:def-invariant_start_and_end},
  we have $\mathstart\not\mathdeprel{\Gamma_{i-1}}\mathend$.
  By \cref{lemma:eq_rel} \cref{item:lemma-eq_rel-deprel},
  we have $\mathstart\not\mathdeprel{\Gamma_{i}}\mathend$.

  \cref{item:def-invariant_existence_of_premise} 
  Let $\mathindFromSTo{t}\in\Delta_{i}$.
  By the definition of $\Delta$, there exists a term $\hat{t}$ such that
  $t \equiv \hat{t}\mathsubstbox{\mathsubst{v_{1}}{u_{2}}, \mathsubst{v_{2}}{u_{1}}}$.
  Then,
  $\mathindFromSTo{\hat{t}\mathsubstbox{\mathsubst{v_{1}}{u_{1}}, \mathsubst{v_{2}}{u_{2}}}}\in\Delta_{i-1}$.
  By the induction hypothesis \cref{item:def-invariant_existence_of_premise},
  $\hat{t}\mathsubstbox{\mathsubst{v_{1}}{u_{1}}, \mathsubst{v_{2}}{u_{2}}}\not\mathdeprel{\Gamma_{i-1}}\mathstart$.
  By \cref{lemma:eq_rel} \cref{item:lemma-eq_rel-deprel},
  $\hat{t}\mathsubstbox{\mathsubst{v_{1}}{u_{2}}, \mathsubst{v_{2}}{u_{1}}}\not\mathdeprel{\Gamma_{i}}\mathstart$.
  Thus, $t\not\mathdeprel{\Gamma_{i}}\mathstart$.

  \cref{item:def-invariant_defined_index} 
  Assume $\mathnext^{n}\mathstart\mathequivrel{\Gamma_{i}}\mathnext^{m}\mathstart$.
  By \cref{lemma:eq_rel} \cref{item:lemma-eq_rel-equiv},
  we have $\mathnext^{n}\mathstart\mathequivrel{\Gamma_{i-1}}\mathnext^{m}\mathstart$ 
  By the induction hypothesis \cref{item:def-invariant_defined_index}, $n=m$.

  %%%%%%%%%%%%%%%%%%%%%%%%%%%%%
  
  Case 4.
  The case \rulename{Case $\mathindToEFromsy$} with the right assumption $\Gamma_{i}\fCenter\Delta_{i}$.
  
  Let $\mathindToEFrom{t}$ be the principal formula of the rule.
  There exists $\Pi$ such that
  $\Gamma_{i-1} \equiv \mleft(\Pi, \mathindToEFrom{t}\mright)$ and
  $\Gamma_{i} \equiv \mleft(\Pi, t = x, \mathindToEFrom{\mathnext x}\mright)$
  for a fresh variable $x$.

  \cref{item:def-invariant_start_and_end}
  If $\mathstart\mathdeprel{\Gamma_{i}}\mathend$,
  then we have $\mathstart\mathdeprel{\Gamma_{i-1}}\mathend$ by \cref{lemma:right_asp}.
  It contradicts the induction hypothesis  \cref{item:def-invariant_start_and_end}.
  Thus, $\mathstart\not\mathdeprel{\Gamma_{i}}\mathend$.

  \cref{item:def-invariant_existence_of_premise}
  Let $\mathindFromSTo{t'}\in\Delta_{i}$.
  If $t'\mathdeprel{\Gamma_{i}}\mathstart$,
  then we have $t'\mathdeprel{\Gamma_{i-1}}\mathstart$ by \cref{lemma:right_asp}.
  It contradicts the induction hypothesis.
  Thus, $t'\not\mathdeprel{\Gamma_{i}}\mathstart$.

  \cref{item:def-invariant_defined_index}
  Assume $\mathnext^{n}\mathstart\mathequivrel{\Gamma_{i}}\mathnext^{m}\mathstart$.
  By \cref{lemma:right_asp},
  $\mathnext^{n}\mathstart\mathequivrel{\Gamma_{i-1}}\mathnext^{m}\mathstart$.
  By the induction hypothesis \cref{item:def-invariant_defined_index}, $n=m$.

  %%%%%%%%%%%%%%%%%%%%%%%%%%%%%%%%%%%%%%%%

  Case 5.
  The case \rulename{Case $\mathindToEFromsy$}
  with the left assumption $\Gamma_{i} \fCenter \Delta_{i}$.
  In this case, $\Gamma_{i-1} \fCenter \Delta_{i-1}$ is a switching point.

  Let $\mathindToEFrom{t}$ be the principal formula of the rule.
  There exists $\Pi$ such that
  $\Gamma_{i-1} \equiv \mleft(\Pi, \mathindToEFrom{t}\mright)$ and
  $\Gamma_{i} \equiv \mleft(\Pi, t = \mathend\mright)$.

  Since $\Gamma_{i-1} \fCenter \Delta_{i-1}$ is a switching point,
  we have $t\not\mathdeprel{\Gamma_{i-1}}\mathstart$.
  By the induction hypothesis \cref{item:def-invariant_start_and_end},
  $\mathstart\not\mathdeprel{\Gamma_{i-1}}\mathend$.

  \cref{item:def-invariant_start_and_end}
  Assume $\mathstart\mathdeprel{\Gamma_{i}}\mathend$ for contradiction.
  By $t\not\mathdeprel{\Gamma_{i-1}}\mathstart$, $\mathstart\not\mathdeprel{\Gamma_{i-1}}\mathend$
  and \cref{lemma:left_asp},
  we have $\mathstart\mathdeprel{\Gamma_{i-1}}\mathend$.
  It contradicts the induction hypothesis \cref{item:def-invariant_start_and_end}.
  Thus, $\mathstart\not\mathdeprel{\Gamma_{i}}\mathend$.

  \cref{item:def-invariant_existence_of_premise}
  Let $\mathindFromSTo{t'}\in\Delta_{i}$.
  Assume $t'\mathdeprel{\Gamma_{i}}\mathstart$ for contradiction.
  By $t\not\mathdeprel{\Gamma_{i-1}}\mathstart$, $\mathstart\not\mathdeprel{\Gamma_{i-1}}\mathend$
  and \cref{lemma:left_asp},
  we have $t'\mathdeprel{\Gamma_{i-1}}\mathstart$.
  It contradicts the induction hypothesis \cref{item:def-invariant_existence_of_premise}.
  Thus, $t'\not\mathdeprel{\Gamma_{i}}\mathstart$.

  \cref{item:def-invariant_defined_index}
  Assume $\mathnext^{n}\mathstart\mathequivrel{\Gamma_{i}}\mathnext^{m}\mathstart$.
  By $t\not\mathdeprel{\Gamma_{i-1}}\mathstart$, $\mathstart\not\mathdeprel{\Gamma_{i-1}}\mathend$
  and \cref{lemma:left_asp},
  we have $\mathnext^{n}\mathstart\mathequivrel{\Gamma_{i-1}}\mathnext^{m}\mathstart$.
  By the induction hypothesis \cref{item:def-invariant_defined_index}, $n=m$.

  %%%%%%%%%%%%%%%%%%%%%%%%%%%%%%%%%%%%%%%%%

  Case 6. The case \rulename{$\mathindFromSTosy$ R${}_{\text{2}}$}.
  Let $\mathindFromSTo{\mathnext t}$ be the principal formula of the rule.

  \cref{item:def-invariant_start_and_end}
  By the induction hypothesis  \cref{item:def-invariant_start_and_end},
  we have $\mathstart\not\mathdeprel{\Gamma_{i-1}}\mathend$.
  Since $\Gamma_{i-1}\equiv\Gamma_{i}$, we have $\mathstart\not\mathdeprel{\Gamma_{i}}\mathend$.

  \cref{item:def-invariant_existence_of_premise}
  Let $\mathindFromSTo{t'}\in\Delta_{i}$.
  Define $\hat{t}$ as $\mathnext t$ if $t'\equiv t$ and $t'$ otherwise.
  By the induction hypothesis \cref{item:def-invariant_existence_of_premise},
  we have $\hat{t}\not\mathdeprel{\Gamma_{i}}\mathstart$.
  Since $\Gamma_{i-1}\equiv\Gamma_{i}$, we have $\hat{t}\not\mathdeprel{\Gamma_{i}}\mathstart$.
  Then, $t'\not\mathdeprel{\Gamma_{i}}\mathstart$.

  \cref{item:def-invariant_defined_index}
  Assume $\mathnext^{n}\mathstart\mathequivrel{\Gamma_{i}}\mathnext^{m}\mathstart$.
  Since $\Gamma_{i-1}\equiv\Gamma_{i}$, 
  we have $\mathnext^{n}\mathstart\mathequivrel{\Gamma_{i-1}}\mathnext^{m}\mathstart$.
  By the induction hypothesis \cref{item:def-invariant_defined_index}, $n=m$.

  \subsection{The proof of \cref{lemma:index}}
  \label{appendix:lemma-index}
  We show \cref{lemma:index}. 

  Let $\tau_{k} \equiv \mathindToEFrom{t_{k}}$.

  \noindent \cref{item:lemma-index_neg}
  It suffices to show that $t_{k+1} \not\mathdeprel{\Gamma_{p+k+1}} \mathstart$ holds
  if $t_{k} \not\mathdeprel{\Gamma_{p+k}} \mathstart$.
  We consider cases according to the rule with the conclusion $\Gamma_{p+k}\fCenter\Delta_{p+k}$.
  
  Case 1. 
  If the rule is \rulename{Weak}, we have the statement by $\Gamma_{p+k+1}\subseteq\Gamma_{p+k}$.

  Case 2. 
  If the rule is \rulename{Subst},
  we have the statement by \cref{lemma:subst_rel} \cref{item:lemma-subst_rel-deprel}.
  
  Case 3. 
  If the rule is \ruleEqLa, 
  then we have the statement by \cref{lemma:eq_rel} \cref{item:lemma-eq_rel-deprel}.
  
  Case 4. The case \rulename{Case $\mathindToEFromsy$} with 
  the right assumption $\Gamma_{p+k+1}\fCenter\Delta_{p+k+1}$. 

  Let $\mathindToEFrom{t}$ be the principal formula of the rule.
  There exists $\Pi$ such that 
  $\Gamma_{p+k}\equiv {\mleft(\Pi, \mathindToEFrom{t} \mright)}$ and 
  $\Gamma_{p+k+1}\equiv {\mleft(\Pi, {t = x}, \mathindToEFrom{\mathnext x}\mright)}$
  with a fresh variable $x$.
  
  We prove this case by contrapositive.
  To show $t_{k} \mathdeprel{\Gamma_{p+k}} \mathstart$,
  assume $t_{k+1} \mathdeprel{\Gamma_{p+k+1}} \mathstart$.
  Define $\hat{t}$ as $t$ if $t_{k+1}\equiv \mathnext x$ and $t_{k+1}$ otherwise. 
  Since $t_{k+1} \mathdeprel{\Gamma_{p+k+1}} \mathstart$ holds, we have $\hat{t} \mathdeprel{\Gamma_{p+k+1}} \mathstart$.
  By \cref{lemma:right_asp}, $\hat{t} \mathdeprel{\Gamma_{p+k}} \mathstart$.
  By $t_{k}\equiv \hat{t}$, we have $t_{k} \mathdeprel{\Gamma_{p+k}} \mathstart$.

  %%%%%%%%%%%%%%%%%%%%%%%%%%%%%%%%%
  Case 5. The case \rulename{Case $\mathindToEFromsy$} 
  with the left assumption $\Gamma_{p+k+1}\fCenter\Delta_{p+k+1}$. 
  In this case,  $\Gamma_{p+k}\fCenter\Delta_{p+k}$ is a switching point.
  
  Let $\mathindToEFrom{t}$ be the principal formula of the rule.
  There exists $\Pi$ such that
  $\Gamma_{p+k} \equiv {\mleft( \Pi, \mathindToEFrom{t}\mright)}$ and
  $\Gamma_{p+k+1} \equiv {\mleft( \Pi, {t = \mathend} \mright)}$. 

  We prove this case by contrapositive.
  To show $t_{k} \mathdeprel{\Gamma_{p+k}} \mathstart$,
  assume $t_{k+1} \mathdeprel{\Gamma_{p+k+1}} \mathstart$.
  Since $\Gamma_{p+k} \fCenter \Delta_{p+k}$ is a switching point,
  we have $t\not\mathdeprel{\Gamma_{p+k}}\mathstart$.
  Since $\Gamma_{p+k} \fCenter \Delta_{p+k}$ is a root-like sequent,
  we have $\mathstart\not\mathdeprel{\Gamma_{p+k}}\mathend$.
  By \cref{lemma:left_asp}, we see that $t_{k} \mathdeprel{\Gamma_{p+k}} \mathstart$.
  
  %%%%%%%%%%%%%%%%%%%%%%%%%%%%%%%%%%

  Case 6. 
  The case (\rulename{$\mathindFromSTosy$ R${}_{\text{2}}$}).

  In this case, since $\Gamma_{p+k}$ is the same as $\Gamma_{p+k+1}$, we have the statement. 

  %%%%%%%%%%%%%%%%%%%%%%%%%%%%%%%%%%%%%

  \noindent \cref{item:lemma-index_regress}
  Let $d_{k}=n$.
  
  Case 1. The case \rulename{Weak}.
  
  If $t_{k+1} \not\mathdeprel{\Gamma_{p+k+1}} \mathstart$, then $d_{k+1}=\bot$.

  Assume $t_{k+1} \mathdeprel{\Gamma_{p+k+1}} \mathstart$.
  By \cref{def:deprel}, there exist $m$, $l\in\mathnat$ such that
  $\mathnext^{m_{0}}t_{k+1}\mathequivrel{\Gamma_{p+k+1}} \mathnext^{m_{1}} \mathstart$.
  By $\Gamma_{p+k+1}\subseteq\Gamma_{p+k}$,
  we have $\mathnext^{m_{0}}t_{k+1}\mathequivrel{\Gamma_{p+k}} \mathnext^{m_{1}} \mathstart$.
  Since $t_{k}\equiv t_{k+1}$, 
  we have $\mathnext^{m_{0}}t_{k}\mathequivrel{\Gamma_{p+k}} \mathnext^{m_{1}} \mathstart$.
  By $d_{k}=n$, we have $m_{1}-m_{0}=n$.
  Thus, $d_{k+1}=n$.

  Case 2. The case \rulename{Subst} with a substitution $\theta$.
  Note that $t_{k}\equiv t_{k+1}\mathsubstbox{\theta}$.
  
  If $t_{k+1} \not\mathdeprel{\Gamma_{p+k+1}} \mathstart$, then $d_{k+1}=\bot$.
  
  Assume that $t_{k+1} \mathdeprel{\Gamma_{p+k+1}} \mathstart$.
  By \cref{def:deprel},
  there exist $m_{0}$, $m_{1}\in\mathnat$ 
  such that $\mathnext^{m_{0}}t_{k+1}\mathequivrel{\Gamma_{p+k+1}} \mathnext^{m_{1}} \mathstart$.
  By \cref{lemma:subst_rel} \cref{item:lemma-subst_rel-equiv},
  $\mathnext^{m_{0}}t_{k+1}\mathsubstbox{\theta}\mathequivrel{\Gamma_{p+k}} \mathnext^{m_{1}} \mathstart$.
  Since $t_{k}\equiv t_{k+1}\mathsubstbox{\theta}$ holds,
  we have $\mathnext^{m_{0}}t_{k}\mathequivrel{\Gamma_{p+k}} \mathnext^{m_{1}}\mathstart$.
  By $d_{k}=n$, we have $m_{1}-m_{0}=n$.
  Thus, $d_{k+1}=n$.

  %%%%%%%%%%%%%%%%%%%%%%%%%%%%%%

  \noindent \cref{item:lemma-index_not_change}
  Let $d_{k}=n$.

  Case 1. The case \ruleEqLa with the principal formula $u_{1} = u_{2}$.
  
  In this case, there exists a term $t$ such that
  $t_{k}\equiv t\mathsubstbox{\mathsubst{v_{1}}{u_{1}}, \mathsubst{v_{2}}{u_{2}}}$
  and $t_{k+1}\equiv t\mathsubstbox{\mathsubst{v_{1}}{u_{2}}, \mathsubst{v_{2}}{u_{1}}}$
  for variables $v_{1}$, $v_{2}$.
  
  By $d_{k}=n$, there exist $m_{0}$, $m_{1}\in\mathnat$ such that
  $\mathnext^{m_{0}} t\mathsubstbox{\mathsubst{v_{1}}{u_{1}}, \mathsubst{v_{2}}{u_{2}}}\mathequivrel{\Gamma_{p+k}} \mathnext^{m_{1}}\mathstart$
  and $m_{1}-m_{0}=n$.
  From \cref{lemma:eq_rel} \cref{item:lemma-eq_rel-equiv},
  $\mathnext^{m_{0}}t\mathsubstbox{\mathsubst{v_{1}}{u_{2}}, \mathsubst{v_{2}}{u_{1}}} \mathequivrel{\Gamma_{p+k+1}} \mathnext^{m_{1}}\mathstart$.
  Thus, $d_{k+1}=m_{1}-m_{0}=n$.

  Case 2. The case \rulename{$\mathindFromSTosy$ R${}_{\text{2}}$}.
  
  Since $\tau_{p+k+1}\equiv\tau_{p+k}$ holds and
  $\Gamma_{p+k}$ is the same as $\Gamma_{p+k+1}$, 
  we have $d_{k+1}=d_{k}$.

  %%%%%%%%%%%%%%%%%%%%%%%%%%%%%%%%%%%%%%%%%%%%%%%

  \noindent \cref{item:lemma-index_Case}
  Let $d_{k}=n$.
  Let $\mathindToEFrom{t}$ be the principal formula of
  the rule \rulename{Case $\mathindToEFromsy$} with the conclusion $\Gamma_{p+k}\fCenter\Delta_{p+k}$.

  \noindent \cref{item:lemma-index_left_ass}
  The case where $\Gamma_{p+k+1}\fCenter\Delta_{p+k+1}$ is the left assumption of the rule.
  In this case, $\Gamma_{p+k}\fCenter\Delta_{p+k}$ is a switching point.
  There exists $\Pi$ such that $\Gamma_{p+k} \equiv {\mleft(\Pi, \mathindToEFrom{t}\mright)}$ and
  $\Gamma_{p+k+1} \equiv {\mleft(\Pi, {t = \mathend}\mright)}$.
  
  By $d_{k}=n$,
  there exist $m_{0}$, $m_{1}\in\mathnat$ such that 
  $\mathnext^{m_{0}}t_{k} \mathequivrel{\Gamma_{p+k}} \mathnext^{m_{1}}\mathstart$ and $m_{1}-m_{0}=n$.
  
  Since the set of formulas with $=$ in $\Gamma_{p+k+1}$
  includes the set of formulas with $=$ in $\Gamma_{p+k}$,
  we have $\mathnext^{m_{0}}t_{k} \mathequivrel{\Gamma_{p+k+1}} \mathnext^{m_{1}}\mathstart$.
  By $\tau_{k}\equiv\tau_{k+1}$, 
  we have $\mathnext^{m_{0}}t_{k+1} \mathequivrel{\Gamma_{p+k+1}} \mathnext^{m_{1}}\mathstart$.
  Thus, $d_{k+1}=m_{1}-m_{0}=n$.

  \noindent \cref{item:lemma-index_not_progress}
  The case where $\Gamma_{p+k+1}\fCenter\Delta_{p+k+1}$ is the right assumption of the rule and
  $\tau_{k}$ is not a progress point of the trace.

  Since $\tau_{k}$ is not a progress point of the trace, we have $\tau_{k+1} \equiv \tau_{k}$.
  By $d_{k}=n$,
  there exist $m_{0}$, $m_{1}\in\mathnat$ such that 
  $\mathnext^{m_{0}}t_{k} \mathequivrel{\Gamma_{p+k}} \mathnext^{m_{1}}\mathstart$ and $m_{1}-m_{0}=n$.

  Since the set of formulas with $=$ in $\Gamma_{p+k}$
  includes the set of formulas with $=$ in $\Gamma_{p+k+1}$,
  we have $\mathnext^{m_{0}}t_{k} \mathequivrel{\Gamma_{p+k+1}} \mathnext^{m_{1}}\mathstart$.
  By $\tau_{k+1}\equiv\tau_{k}$,
  we have $\mathnext^{m_{0}}t_{k+1} \mathequivrel{\Gamma_{p+k+1}} \mathnext^{m_{1}}\mathstart$.
  Thus, $d_{k+1}=m_{1}-m_{0}=n$.

  \noindent \cref{item:lemma-index_progress}
  The case where $\Gamma_{p+k+1}\fCenter\Delta_{p+k+1}$ is the right assumption of the rule and
  $\tau_{k}$ is a progress point of the trace.

  There exists $\Pi$ such that $\Gamma_{p+k}\equiv {\mleft(\Pi, \mathindToEFrom{t} \mright)}$ and
  $\Gamma_{p+k+1}\equiv {\mleft(\Pi, {\mathstart = x}, \mathindToEFrom{\mathnext x}\mright)}$
  for a fresh variable $x$.
  Since $\tau_{k}$ is a progress point of the trace,
  we have $\tau_{k} \equiv \mathindToEFrom{t}$ and
  $\tau_{k+1} \equiv \mathindToEFrom{\mathnext x}$.
  Therefore, $t_{k}\equiv t$ and $t_{k+1}\equiv \mathnext x$.
  By $d_{k}=n$, there exist $m_{0}$, $m_{1}\in\mathnat$ such that
  $\mathnext^{m_{0}}t \mathequivrel{\Gamma_{p+k}} \mathnext^{m_{1}}\mathstart$ and $m_{1}-m_{0}=n$.
  Since the set of formulas with $=$ in $\Gamma_{p+k+1}$
  includes the set of formulas with $=$ in $\Gamma_{p+k}$,
  we have $\mathnext^{m_{0}}t \mathequivrel{\Gamma_{p+k+1}} \mathnext^{m_{1}}\mathstart$.
  By $\mathstart \mathequivrel{\Gamma_{p+k+1}} x$,
  we have $\mathnext^{m_{0}}x \mathequivrel{\Gamma_{p+k+1}} \mathnext^{m_{1}}\mathstart$.
  Hence, $\mathnext^{m} \mathnext x \mathequivrel{\Gamma_{p+k+1}} \mathnext^{m_{1}} \mathnext \mathstart$.
  Therefore, $\mathnext^{m} t_{k+1} \mathequivrel{\Gamma_{p+k+1}} \mathnext^{m_{1}+1} \mathstart$.
  Thus, $d_{k+1}=m_{1}+1-m_{0}=n+1$.


\begin{thebibliography}{99}

%\bibitem{Baelde12}
%  D.~Baelde.
%  \newblock {\em Least and greatest fixed points in linear logic}.
%  \newblock ACM Transactions on Computational Logic,
%  \newblock 13(1), 2012.

\bibitem{Baelde2016}
  D.~Baelde and A.~Doumane and A.~Saurin.
  \newblock {\em Infinitary proof theory: the multiplicative additive case}. 
  \newblock Proceedings of the 25th EACSL Annual Conference on Computer Science Logic (CSL 2016),
  \newblock LIPIcs vol.~62, pp.~42:1--42:17, 2016. 
  
\bibitem{BrotherstonPhD}
  J.~Brotherston.
  \newblock {\em Sequent Calculus Proof Systems for Inductive Definitions}.
  \newblock PhD thesis, University of Edinburgh, 2006.

\bibitem{Brotherston2007}
  J.~Brotherston.
  \newblock {\em Formalised Inductive Reasoning in the Logic of Bunched Implications}.
  \newblock Proceedings of Static Analysis Symposium (SAS'14), 
  \newblock LNCS vol.~4634, pp.~87--103, 2007.
  
\bibitem{Brotherston2011}
  J.~Brotherston and A.~Simpson.
  \newblock {\em Sequent calculi for induction and infinite descent}.
  \newblock Journal of Logic and Computation, 21(6):1177--1216, 2011.

\bibitem{Brotherston11b}
  J.~Brotherston, D.~Distefano, and R.~L.~Petersen. 
  \newblock {\em Automated Cyclic Entailment Proofs in Separation Logic}.
  \newblock Proceedings of the 23rd International Conference on Automated Deduction (CADE-23),
  \newblock LNCS vol.~6803, pp.~131--146, 2011. 
  
\bibitem{Brotherston12}
  J.~Brotherston, N.~Gorogiannis, and R.~L.~Petersen.
  \newblock {\em A Generic Cyclic Theorem Prover}.
  \newblock Proceedings of the 10th Asian Symposium on Programming Languages and Systems (APLAS 2012),
  \newblock LNCS vol.~7705. pp.~350--367, 2012.

\bibitem{Chu15}
  D.~Chu, J.~Jaffar, and M.~Trinh. 
  \newblock {\em Automatic Induction Proofs of Data-Structures in Imperative Programs}.  
  \newblock Proceedings of the 36th ACM SIGPLAN Conference on Programming Language Design and Implementation (PLDI'15),
  \newblock pp.~457--466, 2015. 
  
\bibitem{Cyclist}
  The Cyclist Framework and Provers. 
  \texttt{http://www.cyclist-prover.org/}

\bibitem{DoumanePhD}
  A.~Doumane. 
  \newblock {\em On the infinitary proof theory of logics with fixed points}. 
  \newblock PhD thesis, University of Paris~7, 2017.
  
\bibitem{Fortier2013}
  J.~Fortier and L.~Santocanale.
  \newblock {\em Cuts for circular proofs: semantics and cut-elimination}.
  \newblock Proceedings of Computer Science Logic (CSL 2013),
  \newblock LIPIcs vol.~23, pp.~248--262, 2013.

\bibitem{Kimura2020}
  D.~Kimura, K.~Nakazawa, T.~Terauchi, and H.~Unno.
  \newblock {\em Failure of cut-elimination in cyclic proofs of separation logic}.
  \newblock Computer Software, 37(1):39--52, 2020.

\bibitem{MarinLof1971}
  P.~Martin L\:{o}f. 
  \newblock {\em Haupstatz for the intuitionistic theory of iterated inductive definitions}.
  \newblock Proceedings of the Second Scandinavian Logic Symposium,
  \newblock Studies in Logic and the Foundations of Mathematics, vol.~63, pp.~179–216, 1971.
  
\bibitem{Masuoka2021}
  Y.~Masuoka and M.~Tatsuta.
  \newblock {\em Counterexample to cut-elimination in cyclic proof system for first-order logic with inductive definitions}.
  arXiv: \texttt{https://arxiv.org/abs/2106.11798}, 2021. 
  
\bibitem{Saotome2020}
  K.~Saotome, K.~Nakazawa, and D.~Kimura.
  \newblock {\em Restriction on Cut in Cyclic Proof System for Symbolic Heaps}.
  \newblock Proceedings of Functional and Logic Programming (FLOPS2020), pp.~88--105, 2020.

\bibitem{Saotome2021}
  K.~Saotome, K.~Nakazawa, and D.~Kimura.
  \newblock {\em Failure of cut-elimination in the cyclic proof system of bunched logic with inductive propositions}.
  \newblock Proceedings of the 6th International Conference on Formal Structures for Computation and Deduction (FSCD2021),
  \newblock LIPIcs vol.~195, pp.~11:1--11:14, 2021.

\bibitem{Songbird1}
  Q.~T.~Ta, T.~C.~Le, S.~C.~Khoo, and W.~-N.~Chin. 
  \newblock {\em Automated Mutual Explicit Induction Proof in Separation Logic}.
  \newblock Proceedings of the 21st International Symposium of Formal Methods (FM '16), 
  \newblock LNCS vol.~9995, pp.~659--676, 2016. 
  
\bibitem{Songbird2}
  Q.~T.~Ta, T.~C.~Le, S.~C.~Khoo, and W.~-N.~Chin. 
  \newblock {\em Automated lemma synthesis in symbolic-heap separation logic}.
  \newblock Proceedings of the 45th ACM SIGPLAN Symposium on Principles of Programming Languages (POPL '18),
  \newblock Vol.~2, Article No.9, 2018. 
  
\bibitem{Tatsuta18}
  M.~Tatsuta, K.~Nakazawa, and D.~Kimura.
  \newblock {\em Completeness of cyclic proofs for symbolic heaps with inductive definitions}. 
  \newblock Proceedings of the 17th Asian Symposium on Programming Languages and Systems (APLAS 2019),
  \newblock LNCS vol.~11893, pp.~367--387, 2019.

\bibitem{Tiu2012}
  A.~Tiu and A.~Momigliano. 
  \newblock {\em Cut elimination for a logic with induction and co-induction}.
  \newblock Journal of Applied Logic,
  \newblock vol.~10 (4) pp.~330--367, 2012. 

\bibitem{Tellez2020}
  G.~Tellez and J.~Brotherston. 
  \newblock {\em Automatically Verifying Temporal Properties of Programs with Cyclic Proof}.
  \newblock Journal of Automated Reasoning, vol.~64(3) , pp.~555--578,2020. 

\bibitem{Tsukada2022}
  T.~Tsukada and H.~Unno. 
  \newblock {\em Software Model-Checking as Cyclic-Proof Search}. 
  \newblock arXiv: \texttt{https://arxiv.org/abs/2111.05617} (to appear in POPL 2022)
  
\end{thebibliography}
\end{document}